\documentclass[11pt]{article}

\usepackage[margin=1in]{geometry}
\usepackage{setspace}
\usepackage{amsmath}
\usepackage{amsfonts}
\usepackage{amssymb}
\usepackage{amsthm}
\usepackage{mathtools}

\usepackage{microtype}  

\usepackage{graphicx}
\usepackage{float}
\usepackage{placeins}

\usepackage{booktabs}
\usepackage{caption}
\usepackage{subcaption}

\usepackage[round]{natbib}

\usepackage{algorithm}
\usepackage{algpseudocode}

\usepackage{xcolor}
\usepackage[colorlinks=true,
            linkcolor=blue,
            citecolor=blue,
            urlcolor=blue]{hyperref}

\theoremstyle{plain}
\newtheorem{theorem}{Theorem}
\newtheorem{lemma}{Lemma}

\theoremstyle{definition}

\theoremstyle{remark}
\newtheorem{remark}{Remark}

\title{Multiple Testing of One-Sided Hypotheses with Conservative $p$-values}
\date{}

\author{
Kwangok Seo\thanks{Department of Statistics, Seoul National University, Seoul, Korea}
\and
Johan Lim\thanks{Department of Statistics, Seoul National University, Seoul, Korea}
\thanks{Corresponding author: johanlim@snu.ac.kr}
\and
Hyungwon Choi\thanks{School of Medicine, National University of Singapore, Singapore}
\and
Jaesik Jeong\thanks{Department of Statistics, Chonnam National University, Gwangju, Korea}
}

\begin{document}

\maketitle

\begin{abstract}
We study a large-scale one-sided multiple testing problem in which test statistics follow normal distributions with unit variance, and the goal is to identify signals with positive mean effects. A conventional approach is to compute $p$-values under the assumption that all null means are exactly zero and then apply standard multiple testing procedures such as the Benjamini–Hochberg (BH) or Storey–BH method. However, because the null hypothesis is composite, some null means may be strictly negative. In this case, the resulting $p$-values are conservative, leading to a substantial loss of power.
Existing methods address this issue by modifying the multiple testing procedure itself, for example through conditioning strategies or discarding rules. In contrast, we focus on correcting the $p$-values so that they are exact under the null. Specifically, we estimate the marginal null distribution of the test statistics within an empirical Bayes framework and construct refined $p$-values based on this estimated distribution. These refined $p$-values can then be directly used in standard multiple testing procedures without modification.
Extensive simulation studies show that the proposed method substantially improves power when conventional $p$-values are conservative, while achieving comparable performance to existing methods when conventional $p$-values are exact. An application to phosphorylation data further demonstrates the practical effectiveness of our approach.
\medskip

\end{abstract}

\section{Introduction} 

\subsection{One-Sided Hypothesis Tests with Conservative $p$-Values}
Consider random variables $Z_i \sim N(\mu_i,1)$ for $i \in [m] \coloneqq \{1,2,\ldots,m\}$, 
where $N(\mu,\sigma^2)$ denotes the Gaussian distribution with mean $\mu$ and variance $\sigma^2$.  
We are interested in testing the one-sided hypotheses
\begin{equation*}
    H_{0,i}: \mu_i \le 0
    \quad \text{vs.} \quad
    H_{1,i}: \mu_i > 0,
    \qquad i \in [m].
\end{equation*}
Our aim is to reject as many false null hypotheses as possible while controlling the false discovery rate (FDR, \cite{benjamini1995controlling}) at a prespecified level $q$, based on the observed statistics $\{Z_i\}_{i=1}^m$.\footnote{
Several alternatives to the FDR have been proposed for controlling Type I error in multiple testing, including positive FDR (pFDR; \cite{storey2002direct}), false discovery exceedance (FDX; \cite{genovese2006exceedance}), marginal FDR (mFDR; \cite{sun2007oracle}), sequential goodness-of-fit (SGoF; \cite{martinez2014correlated}), directional FDR (FDR$_{\mathrm{dir}}$; \cite{barber2019knockoff}), and the number of significant effects (NSE; \cite{de2025controlling}), among others. Nevertheless, FDR remains the most widely used criterion for Type I error control, and thus we adopt it throughout this paper.
}
Such one-sided testing problems arise in many applications, including A/B testing \citep{tian2019addis}, qualitative interaction analysis \citep{zhao2019multiple}, and goodness-of-fit assessments for item response theory models \citep{ellis2020gaining}, among others.

A conventional approach is to compute $p$-values as
\begin{equation}\label{conserv-p}
    p_i^{\mathrm{std}} = 1 - \Phi(Z_i), \qquad i \in [m],
\end{equation}
where $\Phi(\cdot)$ denotes the standard Gaussian distribution function.
The resulting $p_i^{\mathrm{std}}$ are then subjected to a multiple testing procedure such as the Benjamini--Hochberg (BH) method \citep{benjamini1995controlling} or the Storey--BH method \citep{storey2002direct,storey2004strong}.

To study the behavior of the $p$-values defined in \eqref{conserv-p}, we introduce the following terminology.
A $p$-value $p$ is called \emph{exact} if, under the null hypothesis, it follows the $\mathrm{Unif}(0,1)$ distribution;  
it is \emph{valid} if it satisfies
\begin{equation*}
    \mathbb{P}(p \le u) \le u, \qquad \forall u \in [0,1];
\end{equation*}
and it is \emph{conservative} if the inequality is strict for some $u \in [0,1]$.

Let $\mathcal{H}_0$ and $\mathcal{H}_1$ denote the index sets of true null and non-null hypotheses, respectively. For $i \in \mathcal{H}_0$, the $p$-value $p_i^{\mathrm{std}}$ defined in \eqref{conserv-p} is always valid but not necessarily exact, since the null hypotheses are composite and may include effects with $\mu_i < 0$.
When $\mu_i < 0$, we have $\mathbb{P}(p_i^{\mathrm{std}} \le u) < u$ for all $u \in (0,1)$, which implies that $p_i^{\mathrm{std}}$ is conservative. Thus, the set $\{p_i^{\mathrm{std}} : i \in \mathcal{H}_0\}$ comprise a mixture of exact $p$-values (arising from $\mu_i=0$) and conservative $p$-values (arising from $\mu_i<0$).

One may ignore the presence of conservative $p$-values and apply BH or Storey--BH. In this case, FDR control remains valid, but the procedure can suffer substantial power loss when conservative $p$-values are abundant.

\subsection{Problem Setting}
Throughout the remainder of the paper, unless stated otherwise, we consider an empirical Bayes setting. 
Specifically, we assume that the observed Gaussian statistics $\{Z_i\}_{i=1}^m$ are generated from the hierarchical model
\begin{equation}\label{hier_model}
    \mu_i \overset{\mathrm{i.i.d.}}{\sim} \pi_0\, M_0(\cdot) + (1-\pi_0)\, M_1(\cdot), \qquad 
    Z_i \mid \mu_i \overset{\mathrm{ind.}}{\sim} N(\mu_i,1), \quad i \in [m],
\end{equation}
where $\pi_0$ denotes the proportion of true null hypotheses, and $M_0$ and $M_1$ denote the distributions of the null and non-null means, respectively, with
$\mathrm{supp}(M_0)=(-\infty,0]$ and $\mathrm{supp}(M_1)=(0,\infty)$.

Under this model, for every $i \in \mathcal{H}_0$, the $p$-value defined in \eqref{conserv-p} is valid for any choice of the null distribution $M_0$:
\begin{equation*}
    \mathbb{P}(p_i^{\mathrm{std}} \le u)
    = \mathbb{E}\!\left[\mathbb{P}(p_i^{\mathrm{std}} \le u \mid \mu_i)\right]
    = \mathbb{E}\!\left[1 - \Phi\!\left(\Phi^{-1}(1-u)-\mu_i\right)\right]
    \overset{(*)}{\le} u, \qquad \forall u \in [0,1],
\end{equation*}
where the last inequality $(*)$ holds because $\mathrm{supp}(M_0) = (-\infty, 0]$.
Consequently, applying the BH or Storey--BH procedure to these $p$-values controls the FDR at the prespecified level $q$.

However, the inequality $(*)$ becomes an equality only when the null distribution $M_0$ degenerates at zero, that is, when all its mass is concentrated at $\mu = 0$. 
Consequently, unless $M_0$ is degenerate at zero, the inequality $(*)$ is strict, and hence $p_i^{\mathrm{std}}$ is conservative for all $i \in \mathcal{H}_0$, leading to the same power loss phenomenon discussed earlier. A key research question is how to alleviate the power loss induced by conservative $p$-values.

\subsection{Existing Methods and Their Applicability}
In this section we briefly review existing methods designed to mitigate the loss of power induced by conservative $p$-values. Before introducing specific procedures, we present two notions introduced in Zhao et al.\ (2019) \cite{zhao2019multiple}: \emph{uniform validity} and \emph{uniform conservativeness}. A valid $p$-value $p$ is said to be uniformly valid if
\begin{equation*}
    \mathbb{P}(p/\tau \le u \mid p \le \tau) \le u,
    \qquad 0 < \tau, u \le 1,
\end{equation*}
and it is uniformly conservative if it is both conservative and uniformly valid.

Tian et al.\ (2019) \cite{tian2019addis} proposed the D-StBH procedure, an extension of the Storey--BH method that additionally incorporates a discarding rule in estimating the null proportion. 
Given user-chosen constants $\lambda < \tau \in (0,1]$ and a target FDR level $q$, they define
\begin{equation*}
    \widehat{\mathrm{FDP}}_{\mathrm{D\text{-}StBH}}(s)
    = \frac{m\hat{\pi}_0 s}{\left(\sum_{i=1}^m \mathbb{I}(p_i^{\mathrm{std}} \le s)\right)\vee 1},
    \qquad
    \hat{\pi}_0
    = \frac{1 + \sum_{i=1}^m \mathbb{I}(\lambda < p_i^{\mathrm{std}} \le \tau)}{m(\tau - \lambda)}.
\end{equation*}
The data-adaptive threshold is then
\begin{equation*}
    \hat{s} = \max\{s \le \tau : 
    \widehat{\mathrm{FDP}}_{\mathrm{D\text{-}StBH}}(s) \le q\},
\end{equation*}
and all hypotheses with $p_i^{\mathrm{std}} \le \hat{s}$ are rejected.
The key distinction from the Storey--BH procedure lies in how the null proportion is estimated.  
D-StBH discards $p$-values larger than $\tau$ when estimating the null proportion, thereby reducing the influence of potentially conservative $p$-values. 
Under mutual independence and uniform validity of the null $p$-values, D-StBH controls the FDR in finite samples.

Another line of work, due to Zhao et al.\ \cite{zhao2019multiple} and Ellis et al.\ \cite{ellis2020gaining}, proposes conditional versions of standard multiple testing procedures. 
The central idea is based on the notion of uniform validity: if a $p$-value $p$ is uniformly valid, then for any threshold $\tau \in (0,1]$, the conditional quantity $p/\tau$ given $p \le \tau$ remains a valid $p$-value. Using this observation, one first restricts attention to the subset of $p$-values below the threshold,
\begin{equation*}
    \{\, p_i^{\mathrm{std}} : p_i^{\mathrm{std}} \le \tau \,\},
\end{equation*}
thereby excluding large $p$-values that are more likely to be conservative. 
The selected $p$-values are then rescaled to restore their support to $[0,1]$, yielding the conditional set
\begin{equation*}
    \mathcal{P}_\tau \coloneqq \{\, p_i^{\mathrm{std}}/\tau : p_i^{\mathrm{std}} \le \tau \,\}.
\end{equation*}
Since each element of $\mathcal{P}_\tau$ is valid under the null hypothesis and the $p$-values remain mutually independent after conditioning, one may apply the BH or Storey--BH procedure to $\mathcal{P}_\tau$ without compromising FDR control.

A key requirement for the methods discussed above is the uniform validity of the null $p$-values. Lemma~\ref{lemma1.1} shows that, under the hierarchical model in \eqref{hier_model}, the $p$-values defined in \eqref{conserv-p} are uniformly valid, and moreover uniformly conservative whenever the null mean distribution $M_0$ is not degenerate at zero. The proof is provided in Appendix~\ref{appen_a}.

\begin{lemma}\label{lemma1.1}
    Let $\{Z_i\}_{i=1}^m$ be generated from the hierarchical model \eqref{hier_model}. For every $i \in \mathcal{H}_0$, the $p$-value $p_i^{\mathrm{std}} = 1 - \Phi(Z_i)$ is uniformly valid for any choice of $M_0$. Furthermore, if the null distribution $M_0$ is not degenerate at zero, then $p_i$ is uniformly conservative.
\end{lemma}

\noindent Lemma~\ref{lemma1.1} establishes that, under the hierarchical model in \eqref{hier_model}, the procedures D-StBH and conditional BH/Storey–BH yield valid FDR control when applied to the $p$-values defined in \eqref{conserv-p}.

\subsection{Our Contributions}
Let $f_0(z) = \int \phi(z; \mu, 1)\, dM_0(\mu)$ and $f_1(z) = \int \phi(z; \mu, 1)\, dM_1(\mu)$,
where $\phi(\cdot;\mu,\sigma^2)$ denotes the density of a $N(\mu,\sigma^2)$ distribution.
The marginal density of $Z_i$ is therefore
\begin{equation*}
    f(z) \coloneqq \pi_0 f_0(z) + (1-\pi_0) f_1(z),
\end{equation*}
and we denote the associated distribution functions by $F_0$, $F_1$, and $F$, respectively.

In the oracle setting where the true null distribution $F_0$ is known, one may define
\begin{equation*}
    p_i^{\mathrm{oracle}} = 1 - F_0(Z_i), \qquad i \in [m],
\end{equation*}
which guarantees that $p_i^{\mathrm{oracle}}$ is exact for every $i \in \mathcal{H}_0$. Consequently, BH or Storey--BH can be applied directly, without any concern about power loss induced by conservative $p$-values.

In practice, however, the true null distribution $F_0$ is unknown. To approximate the oracle procedure, we propose to estimate the null prior $M_0$---and hence $F_0$---in a data-adaptive manner within an empirical Bayes framework. Given an estimate $\hat{F}_0$, we construct the $p$-values
\begin{equation*}
    p_i^{\mathrm{EB}} = 1 - \hat{F}_0(Z_i), \qquad i \in [m],
\end{equation*}
and then apply the BH or Storey--BH procedure to the collection $\{p_i^{\mathrm{EB}}\}_{i=1}^m$ to obtain a rejection set.

Our approach is fundamentally different from existing methods. Procedures such as D\text{-}StBH and conditional BH/Storey--BH mitigate the effect of conservative $p$-values by modifying the multiple testing procedure itself. 
In contrast, we address the problem at its source: by estimating the true null distribution $F_0$, we remove the source of conservativeness in the $p$-values, thereby allowing the BH or Storey--BH procedure to be applied in its original form without modification.

Our main contributions are summarized as follows.

\begin{itemize}
    \item We propose three candidate families for $M_0$:  
    (i) a Dirac delta distribution supported on $(-\infty,0]$;  
    (ii) a Gaussian distribution with mean $0$ and arbitrary variance, truncated to $(-\infty,0]$; and  
    (iii) a finite discrete distribution supported on $(-\infty,0]$.  
    The first serves as a simple baseline, the second provides a natural parametric model for Gaussian means, and the third yields a discretized approximation to a general null distribution.
    
    \item For each specification of the null prior $M_0$, we derive the corresponding null density $f_0$ and develop a data-driven estimation procedure tailored to its structure. A particularly interesting case arises when $M_0$ is modeled as a Gaussian distribution truncated above at zero: in this setting, the induced null density $f_0$ follows a skew-normal distribution with a negative shape parameter---a fact that is far from obvious.

    \item For the cases where the null prior is the Dirac delta distribution and the Gaussian distribution, we establish the consistency of the maximum likelihood estimate (MLE) of the corresponding null density parameter. Building on this result, we further show that the resulting empirical Bayes p-values asymptotically follow the uniform distribution under the null hypothesis. This demonstrates that the proposed empirical Bayes p-values asymptotically inherit the desirable property of the oracle p-values, namely asymptotic exactness under the null hypothesis.

    \item Extensive numerical experiments provide empirical evidence that the proposed method controls the FDR at the nominal level and achieves superior power whenever conservative p-values are present. In settings without conservative nulls, its empirical performance is essentially indistinguishable from that of existing procedures.
    
    \item Finally, we apply our method to the phosphoproteomics data set from a multi-omics study of high-grade serous ovarian cancer (TCGA-HGSC) \citep{zhang2016integrated}, and obtain results consistent with those from the simulation study. Our method yields the largest number of rejections among all compared methods. In particular, it identifies phosphorylation sites that were not detected by the existing methods.
\end{itemize}

\noindent \textbf{Outline.}
The remainder of this paper is organized as follows.
Section~\ref{sec:2} presents the proposed method, including the modeling of the null prior distribution $M_0$ and the estimation of the marginal null distribution $F_0$ tailored to different modeling assumptions.
Section~\ref{sec:3} investigates the finite-sample performance of the proposed method through extensive numerical experiments.
Section~\ref{sec:4} demonstrates the practical utility of the method via an analysis of a phosphoproteomics data set.
Finally, Section~\ref{sec:5} concludes with a discussion.

\section{Methods}\label{sec:2}
Since each $Z_i$ is generated from the hierarchical model in \eqref{hier_model}, $f_0$ belongs to the following class of densities:
\begin{equation*}
    \mathcal{F} \coloneqq \left\{ \int \phi(z ; \mu, 1) \, d\tilde{M}_0(\mu) : \tilde{M}_0 \text{ is a distribution with } \mathrm{supp}(\tilde{M}_0) \subseteq (-\infty, 0] \right\}.
\end{equation*}
A natural approach is to estimate the null density by searching within the class $\mathcal{F}$. However, solving the problem directly over $\mathcal{F}$ is too difficult. Therefore, we adopt the following strategy: we approximate $\mathcal{F}$ by a more tractable class obtained by restricting the distributions allowed for $\tilde{M}_0$, and then estimate the optimal null density within this restricted class based on the observed $Z_i$.

We consider three different specifications for modeling the prior distribution $M_0$. In Section~\ref{subsec_2.1}, we study the case of a Dirac delta distribution with a point mass on $(-\infty, 0]$. In Section~\ref{subsec_2.2}, we examine a Gaussian distribution with mean $0$ and arbitrary variance, truncated to $(-\infty, 0]$. In Section~\ref{subsec_2.3}, we analyze a finite discrete distribution supported on $(-\infty, 0]$. For each specification, we derive the corresponding null density and develop a tailored data-driven estimation procedure. Finally, in Section~\ref{subsec_2.4}, we illustrate our proposed multiple testing procedure based on the estimated null density.

\subsection{Restricting $\tilde{M}_0$ to Dirac Delta Distributions}\label{subsec_2.1} 
In this section, we restrict $\tilde{M}_0$ to a Dirac delta distribution with a point mass on the non-positive real line. This leads to the class
\begin{equation}
\begin{split}\label{class-F-Gaussian}
    \mathcal{F}_{\mathrm{Gaussian}} 
    &\coloneqq \left\{ \int \phi(z ; \mu, 1)\, d\tilde{M}_0(\mu) : 
    \tilde{M}_0 = \delta_{\mu_0},\; \mu_0 \leq 0 \right\}\\
    &\;= \left\{ \phi(z; \mu_0, 1) : \mu_0 \leq 0 \right\},
\end{split}
\end{equation}
where $\delta_x$ denotes the Dirac delta distribution with a point mass at $x$. 

To estimating $f_0$ within the class $\mathcal{F}_{\mathrm{Gaussian}}$, we impose the following assumptions.
\begin{itemize}
    \item (Gaussian Null Assumption) The true null density $f_0$ belongs to $\mathcal{F}_{\mathrm{Gaussian}}$, i.e., $f_0 \in \mathcal{F}_{\mathrm{Gaussian}}$.  
    
    \item[]
    \item (Zero Assumption, \cite{efron2004large}) For a constant $\xi \in \mathbb{R}$, all $Z_i$ with $Z_i \leq \xi$ are generated from the null distribution $F_0$, i.e., $Z_i \sim F_0$ if $Z_i \leq \xi$.  
\end{itemize}
Under these assumptions, the observations $\{Z_i : i \in \mathcal{S}_0 \coloneqq \{i : Z_i \leq \xi\}\}$ are i.i.d.\ samples from a Gaussian distribution with mean $\mu_0 \leq 0$ and unit variance, truncated to $(-\infty, \xi]$: 
\begin{equation*}
    Z_i \sim \frac{\phi(z-\mu_0)}{\Phi(\xi - \mu_0)} \, \mathbb{I}(z \leq \xi).
\end{equation*}

Based on this density, the likelihood function for $\mu_0$ is 
\begin{equation*}
    \mathcal{L}(\mu_0) 
    = \prod_{i \in \mathcal{S}_0} \frac{\phi(Z_i - \mu_0)}{\Phi(\xi - \mu_0)},
\end{equation*}
and the corresponding log-likelihood is 
\begin{equation*}
    \ell(\mu_0) 
    = \sum_{i \in \mathcal{S}_0} \log \phi(Z_i - \mu_0) 
      - |\mathcal{S}_0| \log \Phi(\xi - \mu_0),
\end{equation*}
where $|A|$ denotes the cardinality of the set $A$.
It is important to note that the log-likelihood $\ell(\mu_0)$ is strictly concave in $\mu_0$, which guarantees the uniqueness of the solution that maximizes the log-likelihood; that is, the MLE is unique. However, there is no closed-form expression for the MLE, and thus it must be obtained numerically.

We employ the Newton–Raphson algorithm to compute the MLE. Specifically, the first derivative of $\ell(\mu_0)$ is
\begin{equation*}
    \ell'(\mu_0) 
    = |\mathcal{S}_0|\left\{ (\bar{Z} - \mu_0) + R(\xi - \mu_0)\right\},
\end{equation*}
and the second derivative is
\begin{equation*}
    \ell''(\mu_0) 
    = -|\mathcal{S}_0|\left[ 1 + R(\xi - \mu_0)\big\{(\xi - \mu_0) + R(\xi - \mu_0)\big\}\right],
\end{equation*}
where $\bar{Z} = |\mathcal{S}_0|^{-1}\sum_{i \in \mathcal{S}_0} Z_i$ is the sample mean of the truncated data and $R(x) = \phi(x)/\Phi(x)$ denotes the Mills ratio. The Newton–Raphson update for $\mu_0$ is then given by
\begin{align*}
    \mu_0^{(t+1)} 
    &= \mu_0^{(t)} - \frac{\ell'(\mu_0^{(t)})}{\ell''(\mu_0^{(t)})} \\
    &= \mu_0^{(t)} - 
    \frac{\bar{Z} - \mu_0^{(t)} + R(\xi - \mu_0^{(t)})}
    {-\left[1 + R(\xi - \mu_0^{(t)})\{(\xi - \mu_0^{(t)}) + R(\xi - \mu_0^{(t)})\}\right]}.
\end{align*}
The iteration continues until either the maximum number of iterations $T$ is reached or the parameter update becomes smaller than a prespecified tolerance $\epsilon$,
i.e., $|\mu_0^{(t+1)} - \mu_0^{(t)}| < \epsilon$.

Let $\hat{\mu}_{0,\mathrm{NR}}$ denote the Newton--Raphson estimate of $\mu_0$. Since the parameter space is restricted to $\mu_0 \leq 0$ and the log-likelihood is strictly concave, the final estimator is given by
\begin{equation*}
    \hat{\mu}_0 = \min\{\hat{\mu}_{0,\mathrm{NR}},\,0\}.    
\end{equation*}
A concise summary of this estimation procedure is presented in Algorithm~\ref{alg:1}.

Theorem~\ref{theorem2.1} establishes the consistency of the MLE of $\mu_0$. The proof of Theorem~\ref{theorem2.1} is provided in Appendix~\ref{appen_c}.

\begin{theorem}\label{theorem2.1}
Let $Z_1,\dots,Z_n$ be i.i.d. random samples from the truncated normal density
\[
f_{\mu_0}^T(z)
=
\frac{\phi(z-\mu_0)}{\Phi(\xi-\mu_0)}\mathbb{I}(z\le\xi),
\]
where $\xi\in\mathbb{R}$ is known and $\mu_0$ is the true parameter. Let $\Theta\subset\mathbb{R}$ be a compact parameter space. Suppose that $\mu_0\in\Theta$. Define the log-likelihood criterion
\[
M_n(\mu)
=
\frac{1}{n}\sum_{i=1}^n\log f_\mu^T(Z_i),
\]
and let $\hat{\mu}_n\in\arg\max_{\mu\in\Theta}M_n(\mu)$ be the maximum likelihood estimator. Then,
\[
\hat{\mu}_n\xrightarrow{p}\mu_0 \quad \text{as } n\to\infty.
\]
\end{theorem}
Leveraging the consistency of the MLE of $\mu_0$ together with the Lipschitz continuity of the standard normal C.D.F., we obtain that $p_i^{\mathrm{EB}}$ converges in probability to $p_i^{\mathrm{oracle}}$ under the null hypothesis:
\begin{equation*}
\begin{split}
    |p_i^{\mathrm{EB}} - p_i^{\mathrm{oracle}}| 
    &= \left|\left(1-\Phi(Z_i -\hat{\mu}_n)\right) -  \left(1-\Phi(Z_i - \mu_0)\right)\right|\\
    &= |\Phi(Z_i - \mu_0) - \Phi(Z_i-\hat{\mu}_n)| \le \frac{1}{\sqrt{2\pi}}|\hat{\mu}_n - \mu_0| \xrightarrow{p} 0.
\end{split}
\end{equation*}
Since $p_i^{\mathrm{oracle}}$ follows a uniform distribution under the null hypothesis, it follows that $p_i^{\mathrm{EB}} \xrightarrow{d} \mathrm{Unif}(0,1)$.

\begin{algorithm}[!htb]
\caption{Estimation of $\mu_0$}
\label{alg:1}
\begin{algorithmic}[1]

\State \textbf{Input:} Observations $\{Z_i\}_{i=1}^m$; truncation point $\xi$; tolerance $\varepsilon>0$; maximum iterations $T$.
\State \textbf{Output:} $\hat{\mu}_0$

\State \textbf{Setup:} $\mathcal{S}_0 \gets \{i: Z_i \le \xi\}$,\quad
$\displaystyle \bar{Z} \gets |\mathcal{S}_0|^{-1}\sum_{i\in\mathcal{S}_0} Z_i$,\quad
$R(x) \gets \phi(x)/\Phi(x)$.
\State Initialize $\mu^{(0)} \gets \bar{Z}$ and set $t \gets 0$.

\While{$t < T$}
  \State $\displaystyle \ell' \gets |\mathcal{S}_0|\Big((\bar{Z}-\mu^{(t)}) + R(\xi-\mu^{(t)})\Big)$
  \State $\displaystyle \ell'' \gets -|\mathcal{S}_0|\Big(1 + R(\xi-\mu^{(t)})\{(\xi-\mu^{(t)}) + R(\xi-\mu^{(t)})\}\Big)$
  \State \textbf{Update:} $\mu^{(t+1)} \gets \mu^{(t)} - \ell'/\ell''$
  \If{$|\mu^{(t+1)}-\mu^{(t)}| < \varepsilon$}
     \State \textbf{break}
  \EndIf
  \State $t \gets t+1$
\EndWhile

\State Let $t^\star \gets \min\{t,\,T\}$ 
\State \textbf{Unconstrained estimate:} $\hat{\mu}_{0,\mathrm{NR}} \gets \mu^{(t^\star)}$
\State \textbf{Constrained estimate:} $\hat{\mu}_0 \gets \min\{\hat{\mu}_{0,\mathrm{NR}},\,0\}$

\end{algorithmic}
\end{algorithm}
\subsection{Restricting $\tilde{M}_0$ to Truncated Gaussian Distributions}\label{subsec_2.2}
In this section, we restrict $\tilde{M}_0$ to a Gaussian distribution with mean $0$ and variance $\sigma_0^2 > 0$, truncated to $(-\infty, 0]$. This leads to the class
\begin{equation}\label{class-F-SN}
    \mathcal{F}_{\mathrm{SN}} \coloneqq 
    \left\{ \int \phi(z ; \mu, 1)\, d\tilde{M}_0(\mu) :
    \tilde{M}_0 \sim N(0, \sigma_0^2)\ \text{truncated to } (-\infty, 0],\; \sigma_0 > 0 \right\}.
\end{equation}

Lemma~\ref{lemma2.1} implies that each element of $\mathcal{F}_{\mathrm{SN}}$ lies within the skew-normal family, thereby justifying the use of the subscript `SN' in $\mathcal{F}_{\mathrm{SN}}$. The proof of Lemma~\ref{lemma2.1} is provided in Appendix~\ref{appen_b}. 

\begin{lemma} \label{lemma2.1} 
    For a distribution $\tilde{M}_0$, consider the density
    \begin{equation*}
        f_0(z) = \int \phi(z ; \mu, \sigma^2) \, d\tilde{M}_0(\mu).
    \end{equation*}
    If $\tilde{M}_0$ is a Gaussian distribution with mean $\mu_0 \leq 0$ and variance $\sigma_0^2 > 0$ truncated to $(-\infty, 0]$, then $f_0$ can be written as
    \begin{equation} \label{extend_sn}
         f_0(z) 
         = \frac{1}{\omega \, \Phi(\zeta)} \,
           \phi\!\left( \frac{z-\mu_0}{\omega}\right) 
           \Phi\!\left( \alpha_0(\zeta) + \alpha \frac{z - \mu_0}{\omega} \right),
    \end{equation}
    where $\omega = \sqrt{\sigma^2+\sigma_0^2}, \; 
          \zeta = -\tfrac{\mu_0}{\sigma_0}, \; 
          \alpha = -\tfrac{\sigma_0}{\sigma}$, 
    and $\alpha_0(\zeta) = \zeta\sqrt{\alpha^2 + 1}$.
\end{lemma} 

The distribution in \eqref{extend_sn} can be viewed as an extended version of the skew-normal distribution \cite{azzalini1985class}. When $\mu_0 = 0$, we have $\zeta = 0$ and consequently $\alpha_0(\zeta) = 0$. In this case, the null density $f_0$ simplifies to
\begin{equation*}
    f_0(z) = \frac{1}{2\sqrt{\sigma^2 + \sigma_0^2}} \, 
             \phi\!\left( \frac{z}{\sqrt{\sigma^2 + \sigma_0^2}} \right) 
             \Phi\!\left( -\frac{\sigma_0}{\sigma} \cdot 
             \frac{z}{\sqrt{\sigma^2 + \sigma_0^2}} \right),
\end{equation*}
which corresponds to a skew-normal distribution with location parameter $0$, scale parameter $\sqrt{\sigma^2 + \sigma_0^2}$, and shape parameter $-\sigma_0/\sigma$. 
From this observation, the class $\mathcal{F}_{\mathrm{SN}}$ can be rewritten as
\begin{equation*}
    \mathcal{F}_{\mathrm{SN}} 
    = \left\{ \frac{1}{2\sqrt{1 + \sigma_0^2}} \, 
             \phi\!\left( \frac{z}{\sqrt{1 + \sigma_0^2}} \right) 
             \Phi\!\left( -\sigma_0 \cdot 
             \frac{z}{\sqrt{1 + \sigma_0^2}} \right) : \sigma_0 > 0 \right\}.
\end{equation*}
It is worth noting that the shape parameter $-\sigma_0$ is always negative, indicating that the corresponding density is left-skewed.

We now turn to estimating $f_0$ within the class $\mathcal{F}_{\mathrm{SN}}$. To this end, we replace the Gaussian null assumption with the following:
\begin{itemize}
    \item (Skew\textnormal{-}Normal Null Assumption) The true null density $f_0$ belongs to $\mathcal{F}_{\mathrm{SN}}$, i.e., $f_0 \in \mathcal{F}_{\mathrm{SN}}$.
\end{itemize}
In conjunction with the zero assumption, the observations $\{Z_i : i \in \mathcal{S}_0\}$ are i.i.d.\ samples from a skew-normal distribution with location parameter $0$, scale parameter $\sqrt{1+\sigma_0^2}$, and shape parameter $-\sigma_0$, truncated to $(-\infty,\xi]$:
\begin{equation*}
    Z_i \sim 
    \frac{f_{\mathrm{SN}}(z; 0, \sqrt{1+\sigma_0^2}, -\sigma_0)}
         {F_{\mathrm{SN}}(\xi; 0, \sqrt{1+\sigma_0^2}, -\sigma_0)} \,
         \mathbb{I}(z \le \xi),
\end{equation*}
where $f_{\mathrm{SN}}(\cdot;\mu,\omega,\alpha)$ and $F_{\mathrm{SN}}(\cdot;\mu,\omega,\alpha)$ denote the density and distribution functions of a skew-normal random variable with location $\mu$, scale $\omega$, and shape $\alpha$.
Accordingly, the likelihood for $\sigma_0$ is
\begin{equation*}
    \mathcal{L}(\sigma_0)
    = \prod_{i \in \mathcal{S}_0}
      \frac{f_{\mathrm{SN}}(z_i; 0, \sqrt{1+\sigma_0^2}, -\sigma_0)}
           {F_{\mathrm{SN}}(\xi; 0, \sqrt{1+\sigma_0^2}, -\sigma_0)},
\end{equation*}
and the log-likelihood is
\begin{equation*}
    \ell(\sigma_0)
    = \sum_{i \in \mathcal{S}_0}
      \log f_{\mathrm{SN}}\left(z_i; 0, \sqrt{1+\sigma_0^2}, -\sigma_0\right)
      \;-\; |\mathcal{S}_0| \log F_{\mathrm{SN}}\left(\xi; 0, \sqrt{1+\sigma_0^2}, -\sigma_0\right).
\end{equation*}
Because differentiating $F_{\mathrm{SN}}$ is analytically intractable, we avoid derivative-based methods and instead maximize the log-likelihood using Brent’s method \citep{brent2013algorithms}, as implemented in the \texttt{optimize()} function in \textsf{R}. To enforce $\sigma_0>0$, we reparameterize $\sigma_0=\exp(\eta)$ and optimize over $\eta\in\mathbb{R}$, then recover the estimate as $\hat{\sigma}_0=\exp(\hat{\eta})$. The procedure is summarized in Algorithm~\ref{alg:2}.

Theorem~\ref{theorem2.3} establishes the consistency of the MLE of $\sigma_0$. The proof is given in Appendix~\ref{appen_c}.

\begin{theorem}\label{theorem2.3}
Let $Z_1,\dots,Z_n$ be i.i.d. random samples from the right-truncated skew-normal density
\[
f_{\sigma_0}^T(z)
=
\frac{f_{\sigma_0}(z)}{F_{\sigma_0}(\xi)}\mathbb{I}(z<\xi),
\]
where $\xi\in\mathbb{R}$ is known, $\sigma_0>0$ is the true parameter, and
\[
f_\sigma(z)
=
\frac{2}{\sqrt{1+\sigma^2}}
\phi\!\left(\frac{z}{\sqrt{1+\sigma^2}}\right)
\Phi\!\left(-\frac{\sigma z}{\sqrt{1+\sigma^2}}\right),
\qquad \sigma>0,
\]
and $F_{\sigma}$ denotes the C.D.F. corresponding to $f_{\sigma}$.
Let $\Theta\subset(0,\infty)$ be a compact parameter space. Suppose that $\sigma_0\in\Theta$. Define the log-likelihood criterion
\[
M_n(\sigma)
=
\frac{1}{n}\sum_{i=1}^n\log f_\sigma^T(Z_i),
\quad \sigma\in\Theta,
\]
and let $\hat{\sigma}_n\in\arg\max_{\sigma\in\Theta}M_n(\sigma)$ be the maximum likelihood estimator. Then,
\[
\hat{\sigma}_n\xrightarrow{p}\sigma_0\quad\text{as }n\to\infty.
\]
\end{theorem}

Leveraging the consistency of the MLE of $\sigma_0$, we can establish that $p_i^{\mathrm{EB}}$ converges in probability to $p_i^{\mathrm{oracle}}$ under the null hypothesis:
\begin{equation*}
\begin{split}
    &|p_i^{\mathrm{EB}}-p_i^{\mathrm{oracle}}| 
    = \left| \left(1 - F_{\mathrm{SN}}(Z_i; 0, \sqrt{1+\hat{\sigma}_n^2}, -\hat{\sigma}_n)\right) - \left(1-F_{\mathrm{SN}}(Z_i; 0, \sqrt{1+\sigma_0^2}, -\sigma_0)\right)\right|  \\
    &\quad=\left|F_{\mathrm{SN}}(Z_i; 0, \sqrt{1+\sigma_0^2}, -\sigma_0) - F_{\mathrm{SN}}(Z_i; 0, \sqrt{1+\hat{\sigma}_n^2}, -\hat{\sigma}_n)\right| 
    \leq C|\hat{\sigma}_n - \sigma_0| \xrightarrow{p} 0,
\end{split}
\end{equation*}
where $C$ is a constant depending only on the compact set $\Theta$, and the inequality follows from an intermediate result established in the proof of Lemma~\ref{lemmaC.3} in Appendix~\ref{appen_c}. Since $p_i^{\mathrm{oracle}}$ follows the uniform distribution under the null hypothesis, it follows that $p_i^{\mathrm{EB}} \xrightarrow{d} \mathrm{Unif}(0,1)$.

\begin{algorithm}[!htb]
\caption{Estimation of $\sigma_0$}
\label{alg:2}
\begin{algorithmic}[1]

\State \textbf{Input:} Observations $\{Z_i\}_{i = 1}^m$; truncation point $\xi$; search bounds $\eta_{\min}<\eta_{\max}$.
\State \textbf{Output:} $\hat{\sigma}_0$

\State \textbf{Setup:} $\mathcal{S}_0 \gets \{i : Z_i \le \xi\}$.
\State \textbf{Reparameterization:} $\sigma_0 = \exp(\eta)$ to enforce $\sigma_0>0$.
\State \textbf{Define} the log-likelihood function:
\begin{align*}
\ell(\eta)
&= \sum_{i \in \mathcal{S}_0} 
   \log f_{\mathrm{SN}}\!\left(z_i; 0, \sqrt{1+\exp(2\eta)}, -\exp(\eta)\right) \\
&\quad - |\mathcal{S}_0| 
   \log F_{\mathrm{SN}}\!\left(\xi; 0, \sqrt{1+\exp(2\eta)}, -\exp(\eta)\right).
\end{align*}

\State \textbf{Optimization:} Maximize $\ell(\eta)$ over $[\eta_{\min}, \eta_{\max}]$ 
using Brent’s method as implemented in R’s \texttt{optimize()} function.
\State Let $\hat{\eta}$ denote the maximizer.
\State \textbf{Back-transform:} $\hat{\sigma}_0 \gets \exp(\hat{\eta})$.
\State \Return $\hat{\sigma}_0$

\end{algorithmic}
\end{algorithm}
\subsection{Restricting $\tilde{M}_0$ to Finite Discrete Distribution}\label{subsec_2.3}
Finally, we restrict $\tilde{M}_0$ to be a finite discrete distribution.  
Specifically, for the prespecified grid points $\mu_k \in (-\infty, 0]$,  
$k = 1, 2, \ldots, K$, the null prior distribution $\tilde{M}_0$ is defined as
\begin{equation}\label{finite_disc_mix}
    \tilde{M}_0 = \sum_{k=1}^K p_k \, \delta_{\mu_k},
\end{equation}
where $p_k$ denotes the probability mass assigned to the grid point $\mu_k$.  
Here, both the number of grid points $K$ and their locations $\{\mu_k\}_{k = 1}^K$ are prespecified, so the only unknown parameters are the weights $\{p_k\}_{k=1}^K$.

The null prior distribution in \eqref{finite_disc_mix} induces the following class of finite Gaussian mixture (FGM) densities:
\begin{equation}
\begin{split}\label{class-F-FGM}
    \mathcal{F}_{\mathrm{FGM}}
    &\equiv 
    \mathcal{F}_{\mathrm{FGM}}\left(\{\mu_k\}_{k=1}^K\right) \\
    &\coloneqq
    \left\{
        \int \phi(z ; \mu, 1)\, d\tilde{M}_0(\mu)
        : 
        \tilde{M}_0 = \sum_{k=1}^K p_k \delta_{\mu_k},\;
        \sum_{k=1}^K p_k = 1,\;
        p_k \ge 0
    \right\} \\
    &=
    \left\{
        \sum_{k=1}^K p_k \, \phi(z; \mu_k, 1)
        :
        \sum_{k=1}^K p_k = 1,\;
        p_k \ge 0
    \right\}.
\end{split}
\end{equation}

To estimate $f_0$ within the class $\mathcal{F}_{\mathrm{FGM}}$, we adopt the finite Gaussian mixture null assumption:
\begin{itemize}
    \item (Finite Gaussian Mixture Null Assumption) The true null density $f_0$ belongs to $\mathcal{F}_{\mathrm{FGM}}$, i.e., $f_0 \in \mathcal{F}_{\mathrm{FGM}}$.  
\end{itemize}

\noindent Under this assumption and the zero assumption, the observations  
$\{Z_i : i \in \mathcal{S}_0\}$ are i.i.d.\ from a right-truncated finite Gaussian mixture:
\begin{equation*}
\begin{split}
    Z_i 
    &\sim 
    \frac{\sum_{k = 1}^K p_k \, \phi(z; \mu_k, 1)}
         {\sum_{k = 1}^K p_k \, \Phi(\xi - \mu_k)} \, \mathbb{I}(z \leq \xi) \\
    &= \sum_{k = 1}^K 
       \underbrace{\frac{p_k \, \Phi(\xi - \mu_k)}
       {\sum_{j = 1}^K p_j \, \Phi(\xi - \mu_j)}}_{\displaystyle \eqqcolon \eta_k}
       \cdot 
       \frac{\phi(z; \mu_k, 1)}{\Phi(\xi - \mu_k)} \, \mathbb{I}(z \leq \xi).
\end{split}
\end{equation*}

\noindent The expression above naturally introduces the reparameterization
\[
    \eta_k 
    = \frac{p_k \, \Phi(\xi - \mu_k)}
           {\sum_{j=1}^K p_j \, \Phi(\xi - \mu_j)},
    \qquad k = 1,\ldots,K,
\]
which satisfies $\eta_k \ge 0$ and $\sum_{k=1}^K \eta_k = 1$.  
Hence, the truncated mixture distribution can be viewed as a convex combination of truncated Gaussian densities with mixing proportions $\boldsymbol{\eta}=(\eta_1,\ldots,\eta_K)$.

Since the likelihood depends on $\boldsymbol{p}$ only through the normalized quantities $\boldsymbol{\eta}$, it is convenient to formulate the estimation problem in terms of $\boldsymbol{\eta}$.  
The likelihood and log-likelihood functions can be written as
\begin{equation*}
    \mathcal{L}(\boldsymbol{\eta}) 
    = \prod_{i \in \mathcal{S}_0} 
      \left\{\sum_{k = 1}^K \eta_k \, \frac{\phi(Z_i; \mu_k, 1)}{\Phi(\xi - \mu_k)}\right\},
    \qquad
    \ell(\boldsymbol{\eta}) 
    = \sum_{i \in \mathcal{S}_0} 
      \log \left\{\sum_{k = 1}^K \eta_k \, \frac{\phi(Z_i; \mu_k, 1)}{\Phi(\xi - \mu_k)}\right\}.
\end{equation*}

\noindent A key advantage of this reparameterization is that the log-likelihood is convex in $\boldsymbol{\eta}$.  
Therefore, our estimation proceeds in two steps:  
(i) maximize the convex log-likelihood with respect to $\boldsymbol{\eta}$ using a sequential quadratic programming approach \citep{kim2022semi} to obtain $\hat{\boldsymbol{\eta}}$;  
(ii) recover the original mixing proportions $\hat{\boldsymbol{p}}$ by inverting the mapping from $\boldsymbol{p}$ to $\boldsymbol{\eta}$ at $\hat{\boldsymbol{\eta}}$.  
The complete procedure is summarized in Algorithm~\ref{alg:3}.

\begin{algorithm}[!htb]
\caption{Estimation of $\boldsymbol{p}$}
\label{alg:3}
\begin{algorithmic}[1]

\State \textbf{Input:} Observations $\{Z_i\}_{i=1}^m$; truncation point $\xi$; number of components $K$; support points $\{\mu_k\}_{k=1}^K$ with $\mu_k \leq 0$.
\State \textbf{Output:} $\hat{\boldsymbol{p}} = (\hat{p}_1,\ldots,\hat{p}_K)$

\State \textbf{Setup:} Define the truncated index set $\mathcal{S}_0 \gets \{i: Z_i \leq \xi\}$.  
Reparameterize the mixing proportions by
\[
\eta_k = \frac{p_k \, \Phi(\xi - \mu_k)}{\sum_{j=1}^K p_j \, \Phi(\xi - \mu_j)}, 
\quad k=1,\ldots,K,
\]
so that $\boldsymbol{\eta} = (\eta_1,\ldots,\eta_K)$ lies in the $(K-1)$-dimensional simplex.

\State \textbf{Log-likelihood:}
\[
\ell(\boldsymbol{\eta}) 
= \sum_{i \in \mathcal{S}_0} 
  \log \left\{ \sum_{k=1}^K \eta_k \, 
  \frac{\phi(Z_i;\mu_k,1)}{\Phi(\xi-\mu_k)} \right\}.
\]

\State \textbf{Optimization:} Maximize $\ell(\boldsymbol{\eta})$ over the simplex using sequential quadratic programming \citep{kim2020fast}, obtaining $\hat{\boldsymbol{\eta}}$.

\State \textbf{Back-transformation:}  
For each $k=1,\ldots,K$, compute
\[
\tilde{p}_k \gets \frac{\hat{\eta}_k}{\Phi(\xi-\mu_k)}.
\]
Normalize:
\[
\hat{p}_k \gets \frac{\tilde{p}_k}{\sum_{j=1}^K \tilde{p}_j}, 
\quad k=1,\ldots,K.
\]

\State \Return $\hat{\boldsymbol{p}}$.

\end{algorithmic}
\end{algorithm}

\subsection{Multiple Testing Procedure}\label{subsec_2.4}
Up to this point, we have described how to estimate the null density $f_0$ under the three modeling assumptions introduced in \eqref{class-F-Gaussian}, \eqref{class-F-SN}, and \eqref{class-F-FGM}. 
Given an estimated null density $\hat{f}_0$, we obtain the corresponding null distribution $\hat{F}_0$ and compute $p$-values via
\begin{equation}\label{est-p}
    p_i^{\mathrm{EB}} = 1 - \hat{F}_0(Z_i).
\end{equation}
If $\hat{F}_0$ provides a good approximation to the true null distribution $F_0$, then the resulting $p$-values in \eqref{est-p} are approximately exact.

In practice, however, it is unclear which of the three null modeling assumptions yields the most accurate estimate of the true null density. 
To address this, we first obtain three candidate estimates of $f_0$, each computed using Algorithms~\ref{alg:1}--\ref{alg:3}. 
We then compare these candidates by evaluating their likelihoods on the truncated sample $\{Z_i : i \in \mathcal{S}_0(\xi)\}$, and select as $\hat{f}_0$ the one that achieves the largest value. 
Using this selected null density, we compute $p$-values according to \eqref{est-p}, and subsequently apply a multiple testing procedure—such as the BH or Storey--BH method—to control the FDR. 
The full procedure is summarized in Algorithm~\ref{alg:4}.

\begin{algorithm}[!htb]
\caption{Proposed multiple testing procedure}
\label{alg:4}
\begin{algorithmic}[1]
\State \textbf{Input:} Observed statistics $\{Z_i\}_{i=1}^m$; truncation point $\xi$; target FDR level $q$; algorithm-specific hyperparameters (Algorithms~1--3).
\State \textbf{Output:} Rejection set $\mathcal{R}$.

\State \textbf{Step 1. Estimate candidate null densities.}
\State Run the following estimation algorithms on the truncated sample $\{Z_i : Z_i \le \xi\}$:
\begin{itemize}
    \item Algorithm~\ref{alg:1}: obtain $\hat{f}_0^{\mathrm{G}}$ under the Gaussian null assumption;
    \item Algorithm~\ref{alg:2}: obtain $\hat{f}_0^{\mathrm{SN}}$ under the skew-normal null assumption;
    \item Algorithm~\ref{alg:3}: obtain $\hat{f}_0^{\mathrm{FGM}}$ under the finite Gaussian mixture null assumption.
\end{itemize}

\State \textbf{Step 2. Model selection.}
\State Select the null density estimate
\[
   \hat{f}_0^\ast \in 
   \bigl\{
       \hat{f}_0^{\mathrm{G}},\;
       \hat{f}_0^{\mathrm{SN}},\;
       \hat{f}_0^{\mathrm{FGM}}
   \bigr\}
\]
as the one achieving the largest value among the log-likelihoods evaluated on the truncated sample.

\State \textbf{Step 3. $p$-value computation.}
\For{$i=1,\ldots,m$}
    \State Compute
    \[
       p_i^{\mathrm{EB}}= 1 - \hat{F}_0^\ast(Z_i),
    \]
    where $\hat{F}_0^\ast$ is the c.d.f.\ associated with $\hat{f}_0^\ast$.
\EndFor

\State \textbf{Step 4. Multiple testing.}
\State Apply a multiple testing procedure (e.g., BH or Storey--BH) to the $p$-values $\{p_i^{\mathrm{EB}}\}_{i=1}^m$ at level $q$ to obtain the rejection set $\mathcal{R}$.

\State \Return Rejection set $\mathcal{R}$.
\end{algorithmic}
\end{algorithm}

\begin{remark}
Although our proposed method may look like it fits each model separately and then selects the one with the largest likelihood, the underlying idea is unified: it effectively performs likelihood maximization over $\mathcal{F}_{\mathrm{Gaussian}} \cup \mathcal{F}_{\mathrm{SN}} \cup \mathcal{F}_{\mathrm{FGM}}$. Since each class has a different structure, we use class-specific estimation algorithms, which makes the implementation appear as ``optimize within each class and compare". Conceptually, however, the method amounts to a single optimization over the union of these classes.
\end{remark}

\section{Numerical Study}\label{sec:3}
In this section, we demonstrate our proposed method and compare its performance with existing methods. Section~\ref{subsec_3.1} illustrates the simulation settings considered in this paper. Section~\ref{subsec_3.2} provides a brief summary of the methods for comparison, 
and Section~\ref{subsec_3.3} presents the simulation results.

\subsection{Simulation Settings}\label{subsec_3.1}
Unless otherwise stated, we fix the global simulation settings as follows: the number of hypotheses is set to $m=5{,}000$, the null proportion to $\pi_0 = 0.9$, and the target FDR level to $q=0.1$. For the non-null prior $M_1$, we take a point mass at $3$ so that the corresponding non-null density $f_1$ is the $N(3,1)$ distribution.

We consider two distinct specifications of the null prior distribution $M_0$.
The first specification is a two-point mixture,
\begin{equation}\label{two-comp-mix}
    M_0(\mu)
    = \rho\,\delta_{-1}(\mu) + (1-\rho)\,\delta_{0}(\mu),
    \qquad \rho \in [0,1],
\end{equation}
where the mixing proportion $\rho$ determines the fraction of null effects taking the value $-1$.
When $\rho = 0$, the induced null density $f_0$ coincides with the standard normal density, and hence
$p_i^{\mathrm{std}}$ is exact for all $i \in \mathcal{H}_0$.
For any $\rho > 0$, however, $p_i^{\mathrm{std}}$ is conservative for every $i \in \mathcal{H}_0$, with the degree of conservativeness increasing monotonically in $\rho$.
As a second specification, we consider a left-truncated Gaussian prior,
\begin{equation}\label{truncated-normal}
    M_0(\mu)
    = \frac{\phi(\mu; 0, \sigma_0^2)}{1/2}\,\mathbb{I}(\mu \le 0),
    \qquad \sigma_0 > 0,
\end{equation}
for which Lemma~\ref{lemma2.1} shows that the induced null density $f_0$ is a skew-normal distribution 
with location $0$, scale $\sqrt{1+\sigma_0^2}$, and shape parameter $-\sigma_0$. As $\sigma_0^2$ increases, the density $f_0$ becomes more left-skewed, leading to null $p$-values with a higher degree of conservativeness.

We evaluate type-I error and power using the false discovery rate (FDR) and the true positive rate (TPR):
\begin{align*}
    \mathrm{FDR}
    &=\mathbb E[\mathrm{FDP}], \qquad \mathrm{FDP}=\frac{|\mathcal H_0\cap\mathcal R|}{|\mathcal R|\vee 1},\\
    \mathrm{TPR}
    &=\mathbb E[\mathrm{TPP}],\qquad \mathrm{TPP}=\frac{|\mathcal H_1\cap\mathcal R|}{|\mathcal H_1|\vee 1},
\end{align*}
where $a\vee b=\max\{a,b\}$ and $\mathcal R$ denotes the rejection set.

\subsection{Methods for Comparison}\label{subsec_3.2}
We compare the following four methods:
\begin{itemize}
    \item \textbf{Storey--BH (StBH)} \citep{storey2002direct}:  
This method is a naive approach that ignores the presence of conservative $p$-values: it computes the $p$-values as in \eqref{conserv-p} and subsequently applies the Storey--BH procedure. 
  
    \item \textbf{Conditional Storey--BH (C--StBH)} \citep{zhao2019multiple, ellis2020gaining}:  
    This method employs conditional $p$-values defined as 
    $\{p_i/\tau : i \in \{j : p_j \leq \tau\}\}$, with a pre-specified threshold $\tau$. Throughout our simulations, we set $\tau=0.5$ and apply the Storey--BH procedure to these conditional $p$-values.
    
    \item \textbf{D--StBH} \citep{tian2019addis}:  
    This method employs a discarding rule when estimating the null proportion $\pi_0$. 
    For given constants $0 \leq \lambda < \tau \leq 1$, $\pi_0$ is estimated as
    \begin{equation*}
        \hat{\pi}_0 = \frac{1 + \sum_{i=1}^m \mathbb{I}(\lambda < p_i \leq \tau)}{m(\tau-\lambda)}.
    \end{equation*}
    Following the recommendation of \cite{tian2019addis}, we set $\lambda=0.25$ and $\tau=0.5$ in all simulation scenarios. 
        
    \item \textbf{Proposed Method}:  
    This is our proposed method described in Section~\ref{sec:2}. We employ the Storey--BH procedure for multiplicity correction. Additional implementation details, such as the selection of the truncation point $\xi$ and algorithm-specific hyperparameters, are provided in Appendix~\ref{appen_d}.
\end{itemize}
Since D--StBH is an extension of the Storey--BH procedure, we apply the Storey--BH procedure to all other methods as well to ensure a fair comparison. 

\subsection{Results}\label{subsec_3.3}
\subsubsection{Two-Point Mixture Distribution for $M_0$}
In this section, we consider the two-point mixture distribution described in \eqref{two-comp-mix} as the prior distribution $M_0$. We vary the mixing proportion $\rho$ from 0 to 1 in increments of 0.2. 

Figure~\ref{fig:1} presents the empirical FDR and TPR across different values of $\rho$, each computed by averaging over 1{,}000 replications.
We make the following observations. 
First, all methods control the FDR at the nominal level; our proposed method maintains the FDR tightly at the target level across all values of $\rho$, whereas the other methods become increasingly conservative as $\rho$ approaches 1. 
Second, in terms of TPR, our proposed method consistently outperforms the competing methods whenever $\rho>0$. Moreover, the performance gap widens as $\rho$ increases; in particular, when $\rho=1$, our method achieves substantially higher power than all other methods. Even when $\rho=0$, its power remains comparable to the alternatives. 
Finally, the power of the StBH procedure decreases monotonically as $\rho$ approaches 1. 
This behavior is somewhat counterintuitive, since signal detection would be expected to become easier as $\rho$ increases. This power loss arises because the StBH procedure fails to account for the increasing conservativeness of the null $p$-values.

\begin{figure}[!htb]
    \centering
    \includegraphics[width=0.9\linewidth]{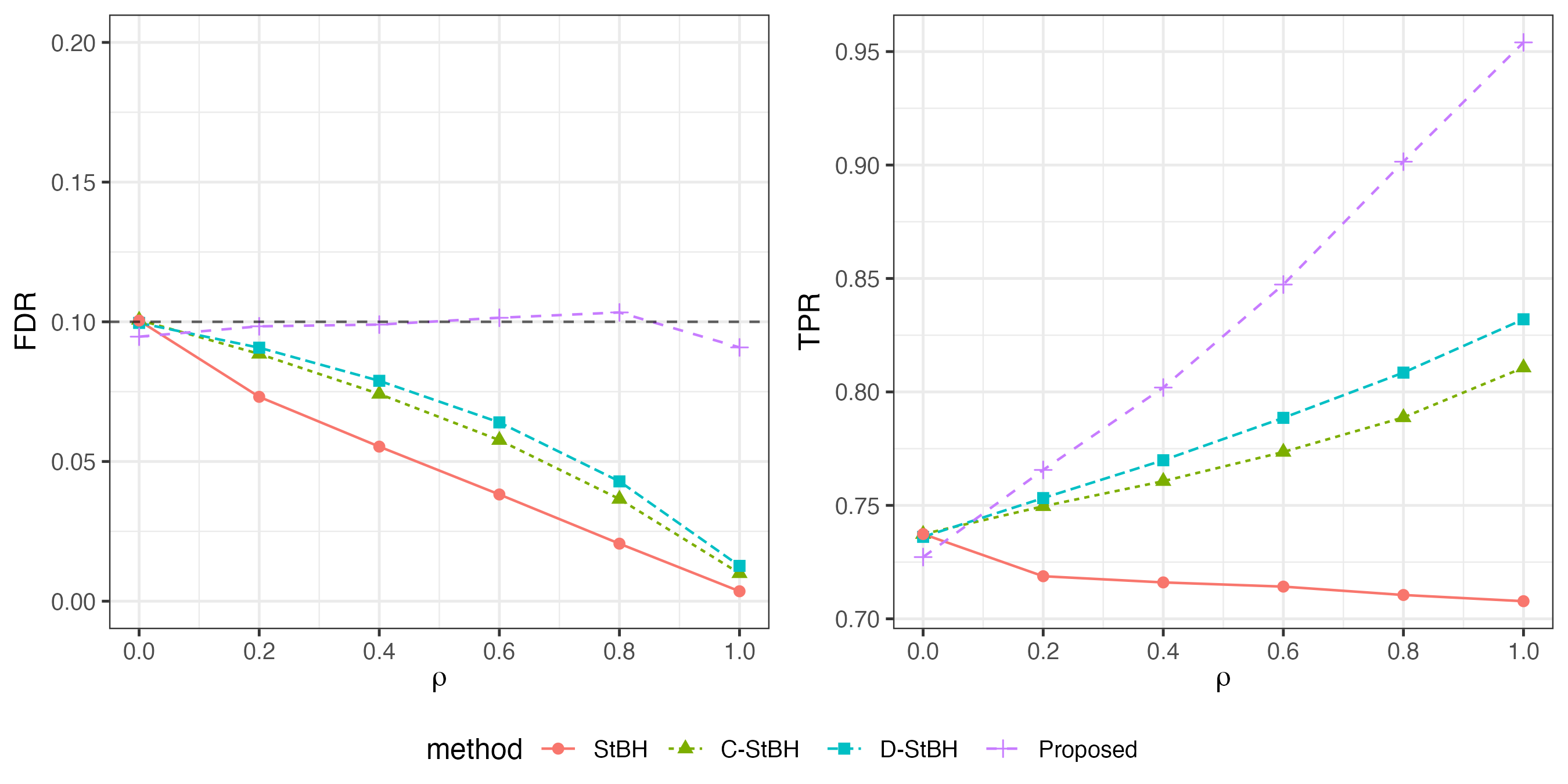}
    \caption{Empirical FDR and TPR averaged over 1,000 replications under the two-point mixture null distribution $M_0$ with varying mixing proportion $\rho$. The dashed line in the left panel indicates the target FDR level of 0.1.}
    \label{fig:1}
\end{figure}

Figure~\ref{fig:2} presents histograms of the $p$-values computed under the standard Gaussian null, $\{p_i^{\mathrm{std}}\}_{i=1}^{5{,}000}$, and under our estimated null distribution, $\{p_i^{\mathrm{EB}}\}_{i=1}^{5{,}000}$, when $\rho=1$. The histogram based on our estimate shows a right tail that is close to uniformly distributed, whereas the histogram under the standard Gaussian null exhibits a marked deviation from uniformity. This demonstrates that our method accurately captures the true null distribution and produces well-calibrated $p$-values.

\begin{figure}[!htb]
    \centering
    \includegraphics[width=0.9\linewidth]{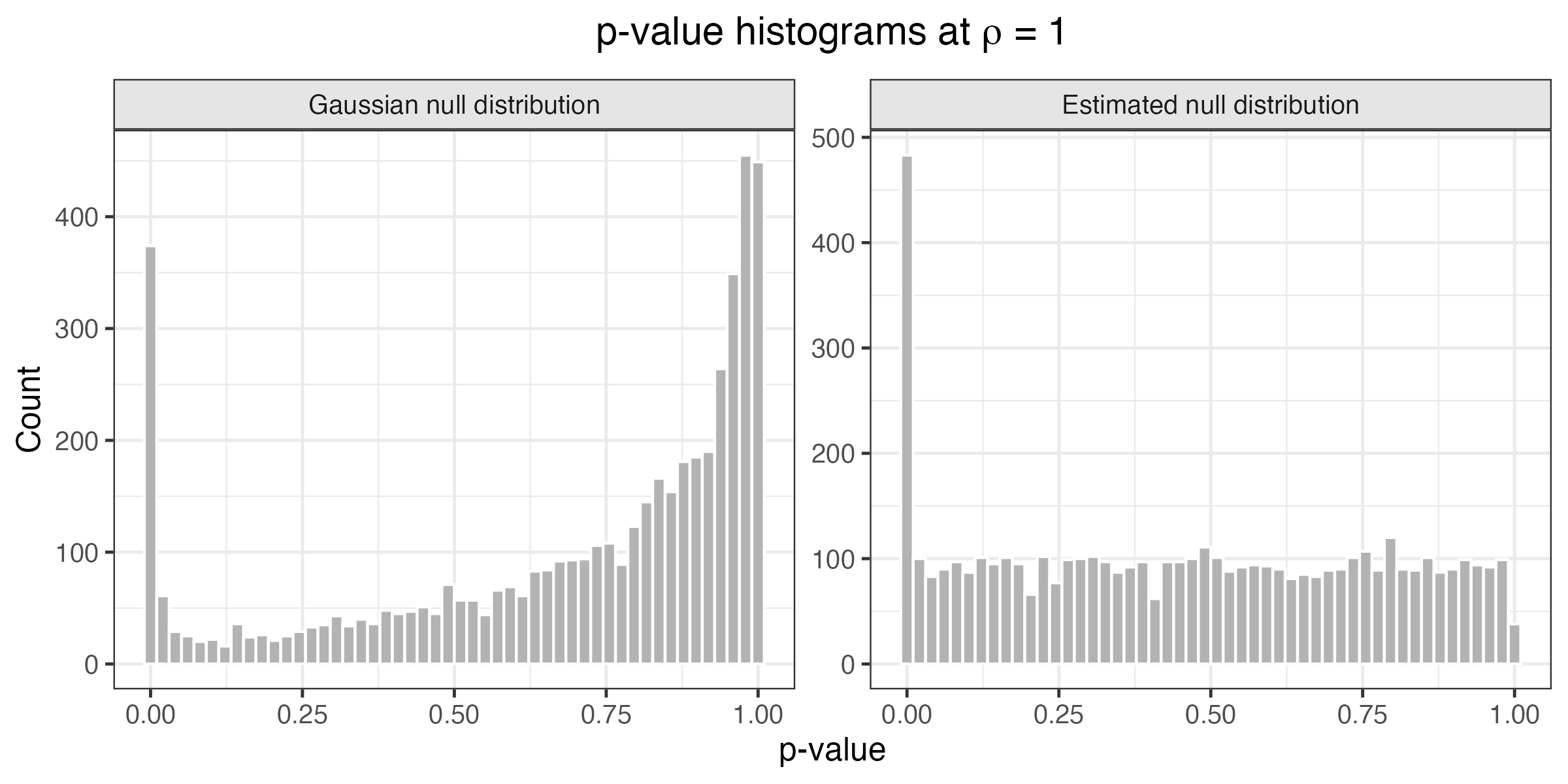}
    \caption{Histograms of $p$-values computed under the standard Gaussian null distribution (left panel) and under our estimated null distribution (right panel) at mixing proportion $\rho = 1$.
}
    \label{fig:2}
\end{figure}
\subsubsection{Truncated Gaussian Distribution for $M_0$}
In this section, we take the truncated Gaussian distribution in \eqref{truncated-normal} as the prior $M_0$ and vary $\sigma_0$ from 1 to 2 in increments of 0.2. 

Figure~\ref{fig:3} reports the empirical FDR and TPR across different values of $\sigma_0$, each computed by averaging over 1{,}000 replications. The overall pattern is similar to Figure~\ref{fig:1}: all methods maintain FDR control at the target level, while the proposed method consistently achieves the highest power across all values of $\sigma_0$, followed by D–StBH, C–StBH, and StBH.  The advantage of the proposed method over StBH is especially pronounced when $\sigma_0=2$, which represents the most conservative scenario among all cases considered. 

\begin{figure}[!htb]
    \centering
    \includegraphics[width=0.9\linewidth]{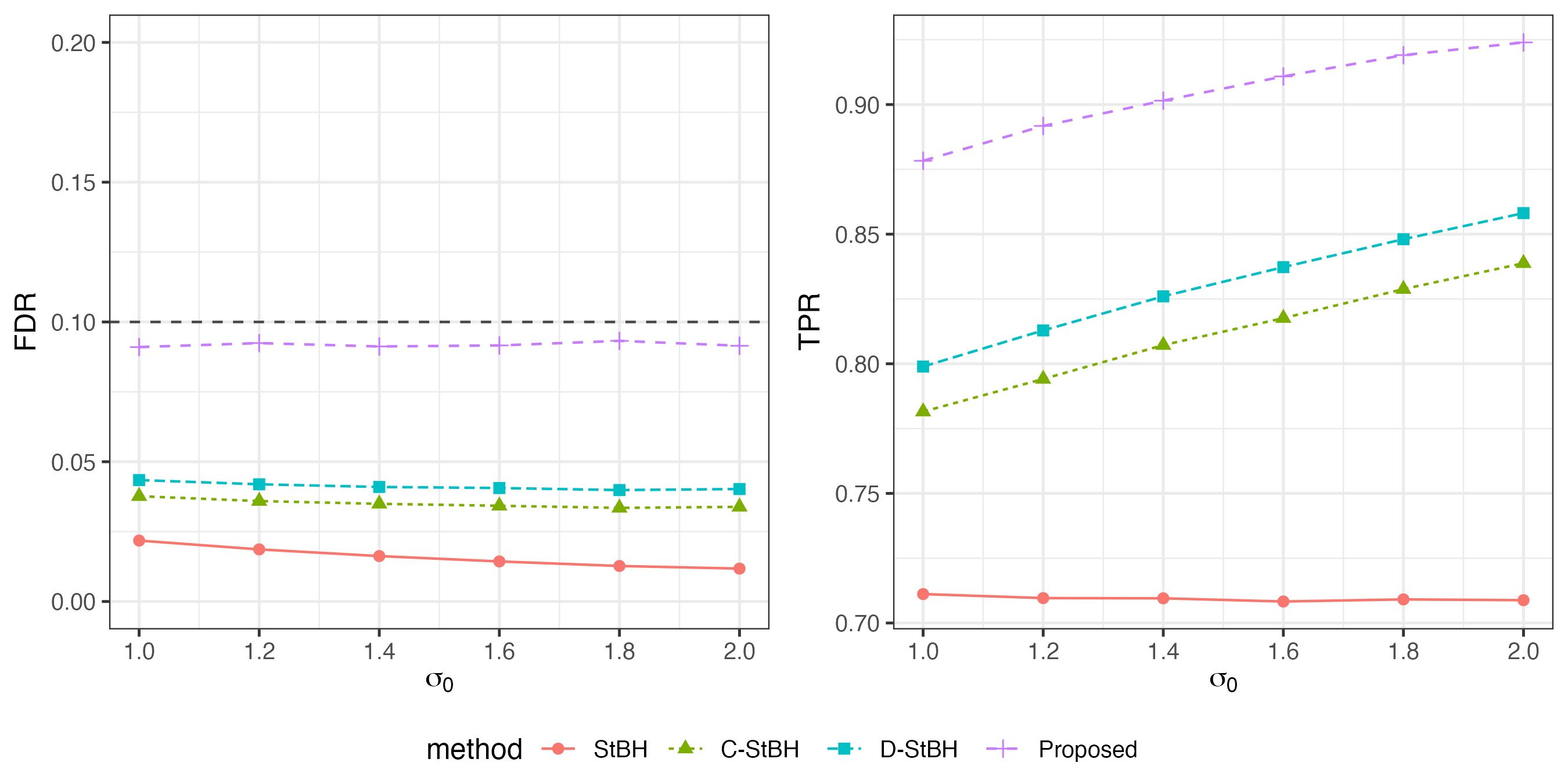}
    \caption{Empirical FDR and TPR averaged over 1,000 replications under the truncated Gaussian null distribution $M_0$ with varying variance $\sigma_0$. The dashed line in the left panel indicates the target FDR level of 0.1.}
    \label{fig:3}
\end{figure}

Figure~\ref{fig:4} presents histograms of the $p$-values computed under the standard Gaussian null, $\{p_i^{\mathrm{std}}\}_{i=1}^{5{,}000}$, and under our estimated null distribution, $\{p_i^{\mathrm{EB}}\}_{i=1}^{5{,}000}$, when $\sigma_0=2$. The $p$-values based on our estimate display a uniform right tail, whereas those under the standard Gaussian null exhibit a clear deviation from uniformity. This indicates that our method accurately recovers the true null distribution and yields well-calibrated $p$-values.

\begin{figure}[!htb]
    \centering
    \includegraphics[width=0.9\linewidth]{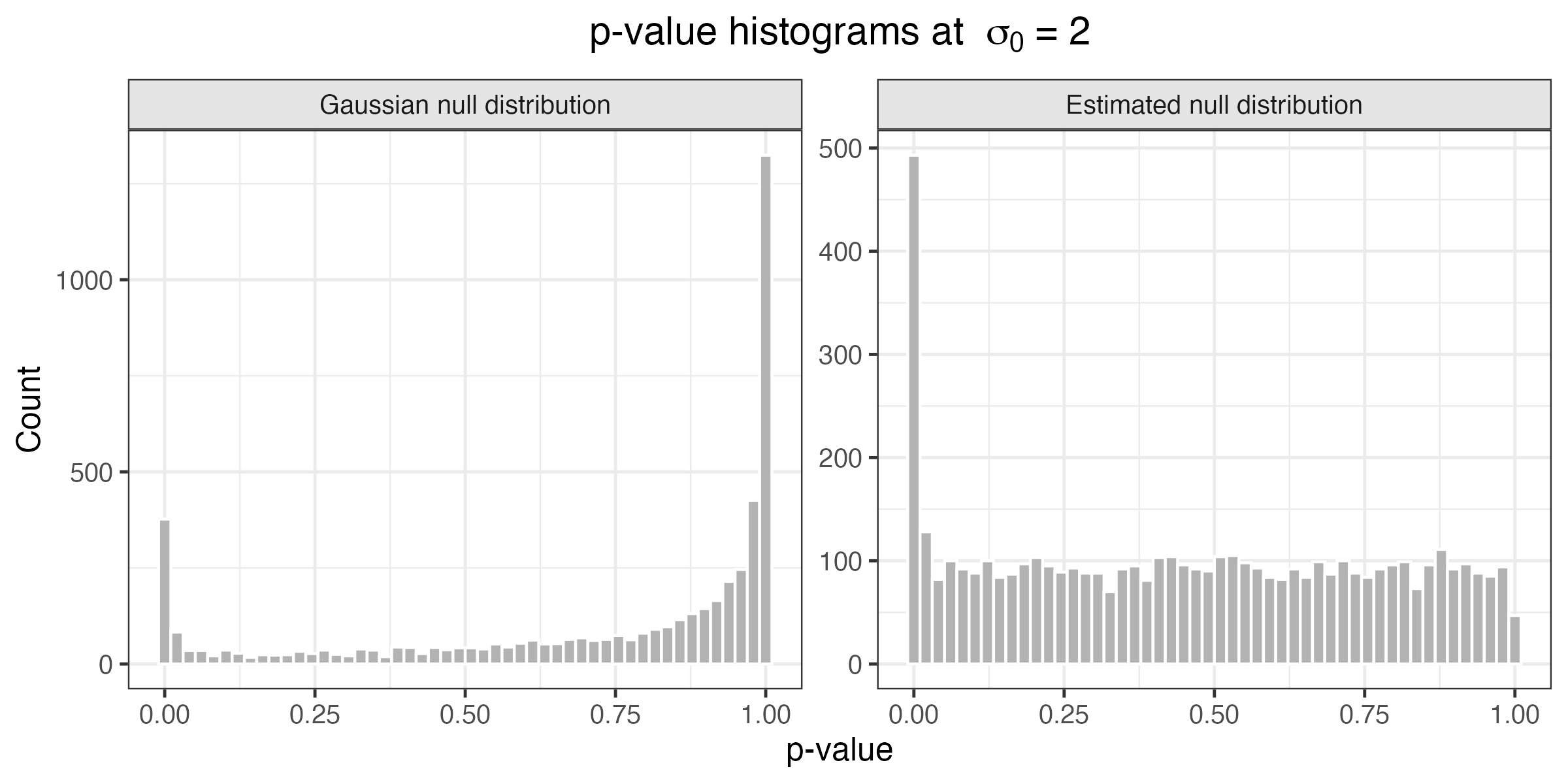}
    \caption{Histograms of $p$-values computed under the standard Gaussian null distribution (left panel) and under our estimated null distribution (right panel) at variance $\sigma_0 = 2$.
}
    \label{fig:4}
\end{figure}

\section{Real Data Example}\label{sec:4}
In this section, we evaluate the proposed methods using the TCGA-HGSC phosphoproteomics data and compare their performance with existing approaches.

The dataset contains normalized phosphorylation abundance measurements for 5{,}746 phosphorylation sites observed across 67 tumor samples of high-grade serous ovarian carcinoma. Each phosphorylation site is defined by a gene symbol and a modification position, and its phosphorylation level was quantified for all available samples.
Following the original study \citep{zhang2016integrated}, the 67 samples are classified into five proteomic subtypes: A (Differentiated), B (Immunoreactive), C (Proliferative), D (Mesenchymal), and E (Stromal), with 21, 10, 7, 21, and 8 samples in each subtype, respectively.

\subsection{Experimental Setup for Subtype Comparisons}
It is customary to regard differential abundance analysis, e.g., in gene expression studies, as a two-sided hypothesis testing problem aimed at detecting overall dysregulation. In phosphorylation analysis, however, the identification of amplified and depleted signaling levels (phosphorylation on serine, threonine, and tyrosine residues) in specific groups of samples carries distinct biological meanings, and therefore separate one-sided tests are more meaningful. 

In this context, we focus on identifying phosphorylation sites that exhibit significantly higher expression levels in one subtype compared to another. Specifically, we consider pairwise comparisons between two subtypes, formulated as a large-scale multiple testing problem of the form
\begin{equation*}
    H_{0,i}: \mu_{s_1,i} \le \mu_{s_2,i} \quad \text{vs.} \quad H_{1,i}: \mu_{s_1, i} > \mu_{s_2, i}, \quad \text{for }i = 1, \ldots, m,
\end{equation*}
where $\mu_{s_1,i}$ and $\mu_{s_2,i}$ denote the mean phosphorylation levels of site $i$ in subtypes $s_1$ and $s_2$, respectively.

Among the 20 possible subtype pairs, we focus on two representative cases for which the sample sizes are sufficiently large and the corresponding null distributions appear to be conservative:
\begin{align*}
    H_{0,i}: \mu_{D,i} \le \mu_{B,i} \quad &\text{vs.} \quad H_{1,i}: \mu_{D,i} > \mu_{B,i} \quad \text{and} \\ 
    H_{0,i}: \mu_{D,i} \le \mu_{A,i} \quad &\text{vs.} \quad H_{1,i}: \mu_{D,i} > \mu_{A,i}, \quad \text{for } i = 1, \ldots, 5{,}746,
\end{align*}
corresponding to comparisons between subtypes D and B, and between D and A, respectively.

Given a pair of subtypes to be compared, we compute the two-sample $t$-statistic for each phosphorylation site and apply all competing methods to obtain their respective rejection sets.
For the existing methods (StBH, C-StBH, and D-StBH), the $p$-values are computed as
\begin{equation*}
    p_i^{\mathrm{std}} = 1 - \Phi(t_i), \qquad i = 1,\ldots,5{,}746,
\end{equation*}
where $t_i$ denotes the $t$-statistic for the $i$-th phosphorylation site. Each method is then applied to these $p$-values to determine the rejection set.
For the proposed method, we first estimate the null distribution $F_0$ from the observed $t$-statistics $\{t_i\}_{i=1}^{5{,}746}$, construct the $p$-values as
\begin{equation*}
    p_i^{\mathrm{EB}} = 1 - \hat{F}_0(t_i), \qquad i = 1,\ldots,5{,}746,
\end{equation*}
and subsequently apply the Storey--BH procedure.
The two subtype comparisons involve sample sizes of 21 + 10 = 31 (D vs.\ B) and 21 + 21 = 42 (D vs.\ A), for which the normal approximation to the $t$-distribution is sufficiently accurate. As a result, computing $p$-values using the standard Gaussian distribution for the existing methods, as well as applying the proposed method to the observed t-statistics $\{t_i\}_{i =1}^{5{,}746}$, is well justified in this setting. Throughout the analysis, the target FDR level is set to $0.1$.

\subsection{Results}
We first present the results for the comparison between subtypes D and B. The numbers of rejected hypotheses were 260, 318, 326, and 400 for the StBH, C-StBH, D-StBH, and the proposed method, respectively. 
The rejection sets were nested as 
\begin{equation*}
    \mathcal{R}_{\text{StBH}} \subset \mathcal{R}_{\text{C-StBH}} \subset \mathcal{R}_{\text{D-StBH}} \subset \mathcal{R}_{\text{Proposed}},
\end{equation*}
indicating that our method additionally rejected 74 hypotheses that were not detected by the existing approaches.

Figure~\ref{fig:5} shows the phosphorylation abundance patterns for the 74 sites uniquely identified by the proposed method. The majority of these sites exhibit higher phosphorylation abundance in subtype D and lower abundance in subtype B, as indicated by the predominance of red and blue colors, respectively, though the strength of this contrast varies across sites. This observation underscores the practical power advantage of the proposed method over existing procedures.

Figure~\ref{fig:6} compares histograms of $p$-values obtained under the standard Gaussian null distribution, $\{p_i^{\mathrm{std}}\}_{i=1}^{5{,}746}$, with those based on the null distribution estimated by our method, $\{p_i^{\mathrm{EB}}\}_{i=1}^{5{,}746}$. The $p$-values obtained from our estimated null exhibit an approximately uniform right tail, whereas those derived from the Gaussian distribution display a pronounced accumulation near 1.

Figure~\ref{fig:7} displays the histogram of the $t$-statistics together with two scaled null density curves.
The red curve corresponds to the null density estimated by our method, scaled by the estimated null proportion obtained via Storey’s method applied to the $p$-values from our procedure.
The blue curve applies the same scaling factor to the standard Gaussian density, ensuring a fair comparison between the two densities. The scaled estimated null density closely matches the empirical distribution of the test statistics, whereas the scaled standard Gaussian density exhibits a visible mismatch. 

\begin{figure}[!htb]
    \centering
    \includegraphics[width=0.65\linewidth]{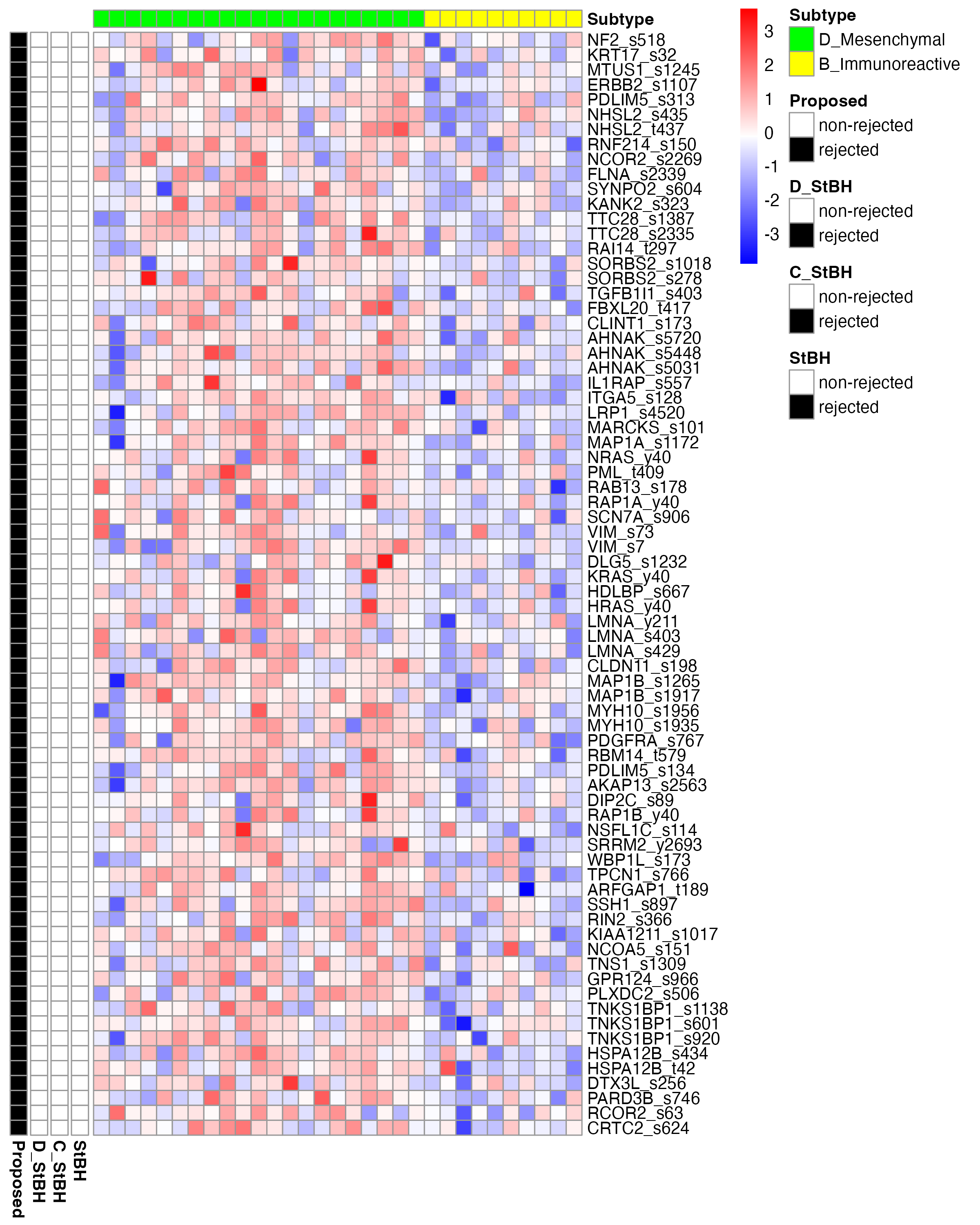}
    \caption{Heatmap of phosphorylation abundance for phosphorylation sites uniquely identified by the proposed method. Rows correspond to phosphorylation sites and columns correspond to samples. Colors are scaled such that increasingly red (blue) shades indicate higher (lower) phosphorylation abundance. Left-side annotations indicates rejection status across methods.}
    \label{fig:5}
\end{figure}

\begin{figure}[!htb]
    \centering
    \includegraphics[width=0.9\linewidth]{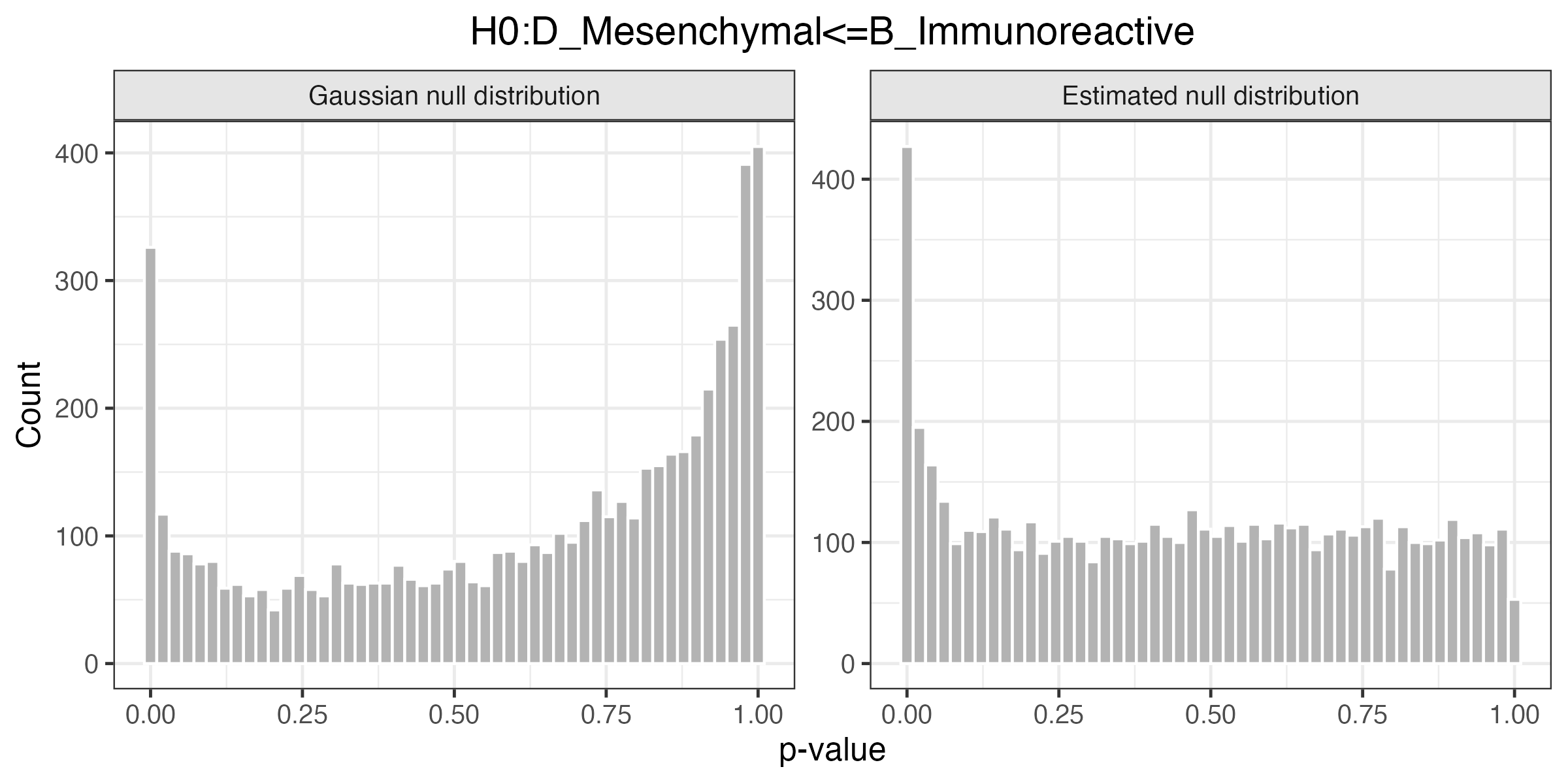}
    \caption{Histograms of the $p$-values computed under the Gaussian null distribution (left) and under our estimated null distribution (right), corresponding to the hypothesis tests $H_{0,i}: \mu_{D,i} \le \mu_{B,i}$.}
    \label{fig:6}
\end{figure}

\begin{figure}[!htb]
    \centering
    \includegraphics[width=0.4\linewidth]{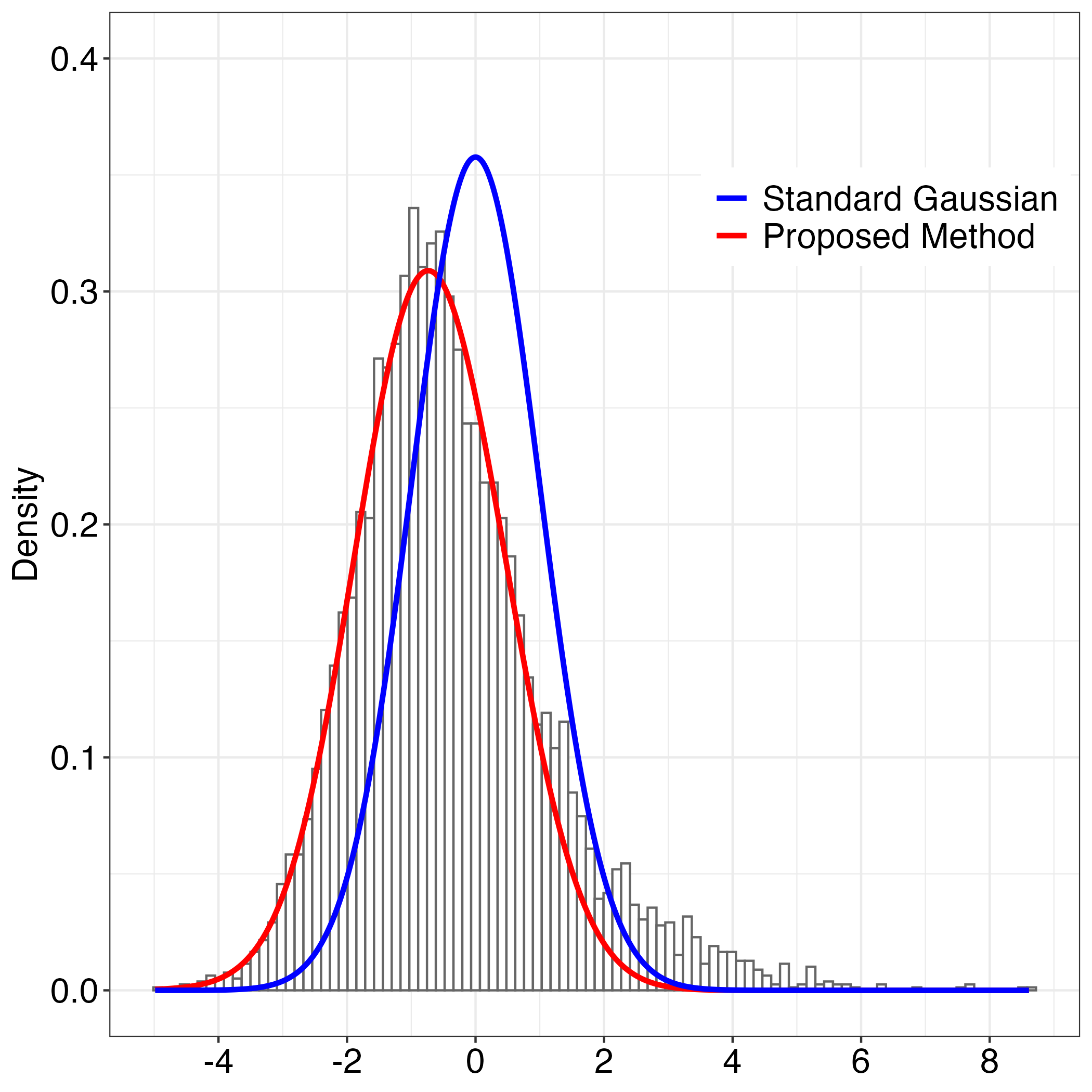}
    \caption{Comparison between the scaled null density estimated by our proposed method (red line) and the standard Gaussian distribution scaled by the same factor (blue line), corresponding to the hypothesis tests $H_{0,i}: \mu_{D,i} \le \mu_{B,i}$.}
    \label{fig:7}
\end{figure}

A similar trend was observed for the comparison between subtypes D and A. In this case, the numbers of rejected hypotheses were 239, 292, 280, and 293 for the StBH, C-StBH, D-StBH, and the proposed method, respectively. The rejection sets followed the inclusion relationship 
\begin{equation*}
    \mathcal{R}_{\text{StBH}} \subset \mathcal{R}_{\text{D-StBH}} \subset \mathcal{R}_{\text{C-StBH}} \subset \mathcal{R}_{\text{Proposed}},
\end{equation*}
where the order of C-StBH and D-StBH was reversed compared to the previous case. Nevertheless, our proposed method again produced the largest number of rejections.

Figure~\ref{fig:8} presents the phosphorylation abundance of site ANKRD50\_s1167, which is uniquely identified by the proposed method. The majority of samples from subtype D exhibit positive phosphorylation abundance, whereas most samples from subtype A exhibit negative abundance.
Consistently, the subtype-wise average phosphorylation abundance is positive for subtype D (approximately $0.553$) and negative for subtype A (approximately $-0.125$), indicating a clear and coherent subtype-specific shift that is effectively captured by the proposed method but missed by existing procedures.

Figure~\ref{fig:9} demonstrates that the $p$-values from our method have an approximately uniform right tail, while the Standard Gaussian distribution based $p$-values exhibit the characteristic heavy right tail that signals the presence of conservative $p$-values. 

In Figure~\ref{fig:10}, the null density estimated by our method closely matches the empirical distribution of the test statistics, whereas the standard Gaussian null shows a clear mismatch.

\begin{figure}[!htb]
    \centering
    \includegraphics[width=0.9\linewidth]{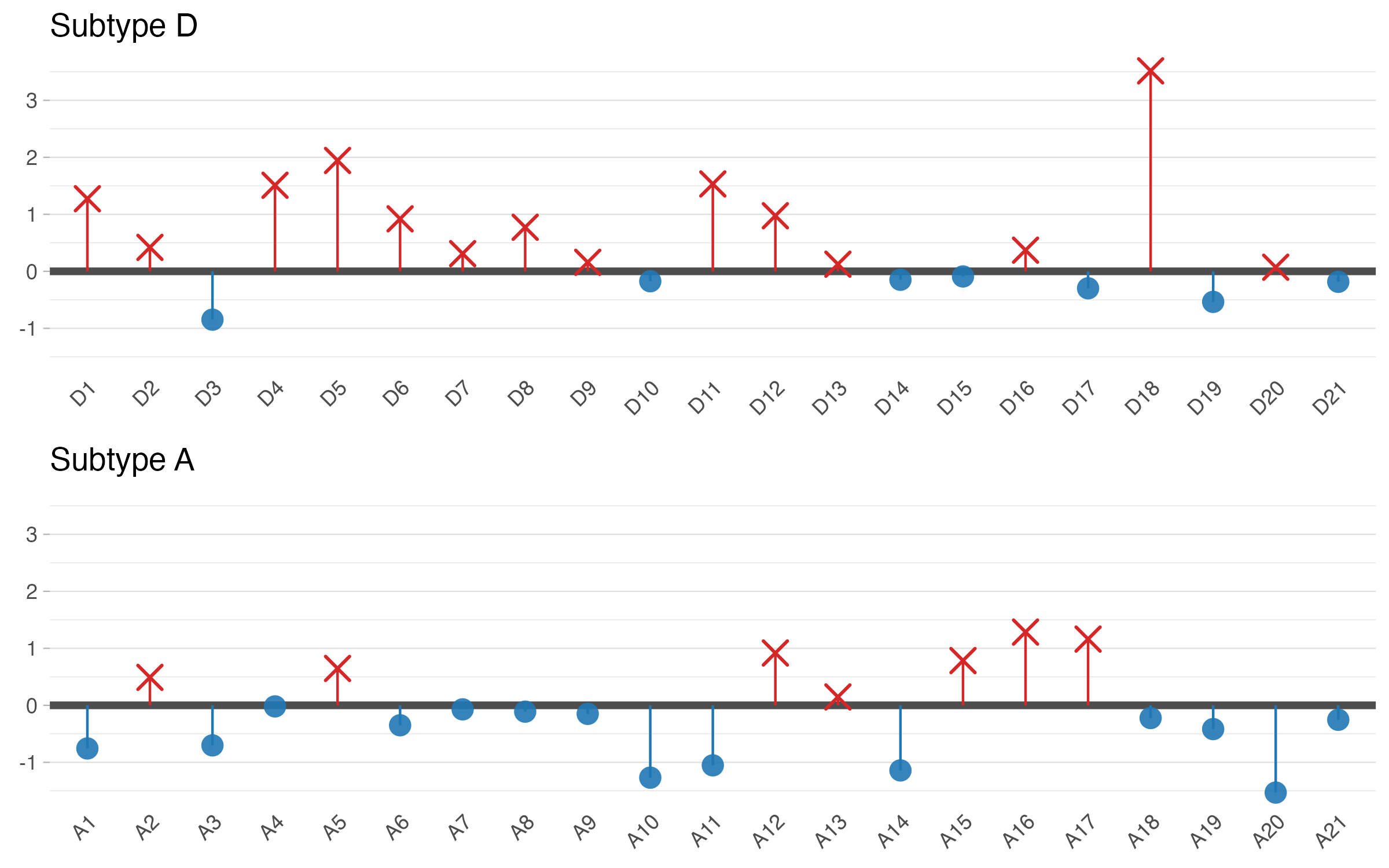}
    \caption{Phosphorylation abundance of site ANKRD50\_s1167 uniquely identified by the proposed method. The top panel shows individual samples from subtype D (D1--D21), and the bottom panel shows individual samples from subtype A (A1--A21). Each vertical segment represents the phosphorylation abundance of an individual sample. Positive values are shown as upward red segments with cross symbols, whereas negative values are shown as downward blue segments with circular symbols. The horizontal black line indicates zero phosphorylation abundance.}
    \label{fig:8}
\end{figure}

\begin{figure}[!htb]
    \centering
    \includegraphics[width=0.9\linewidth]{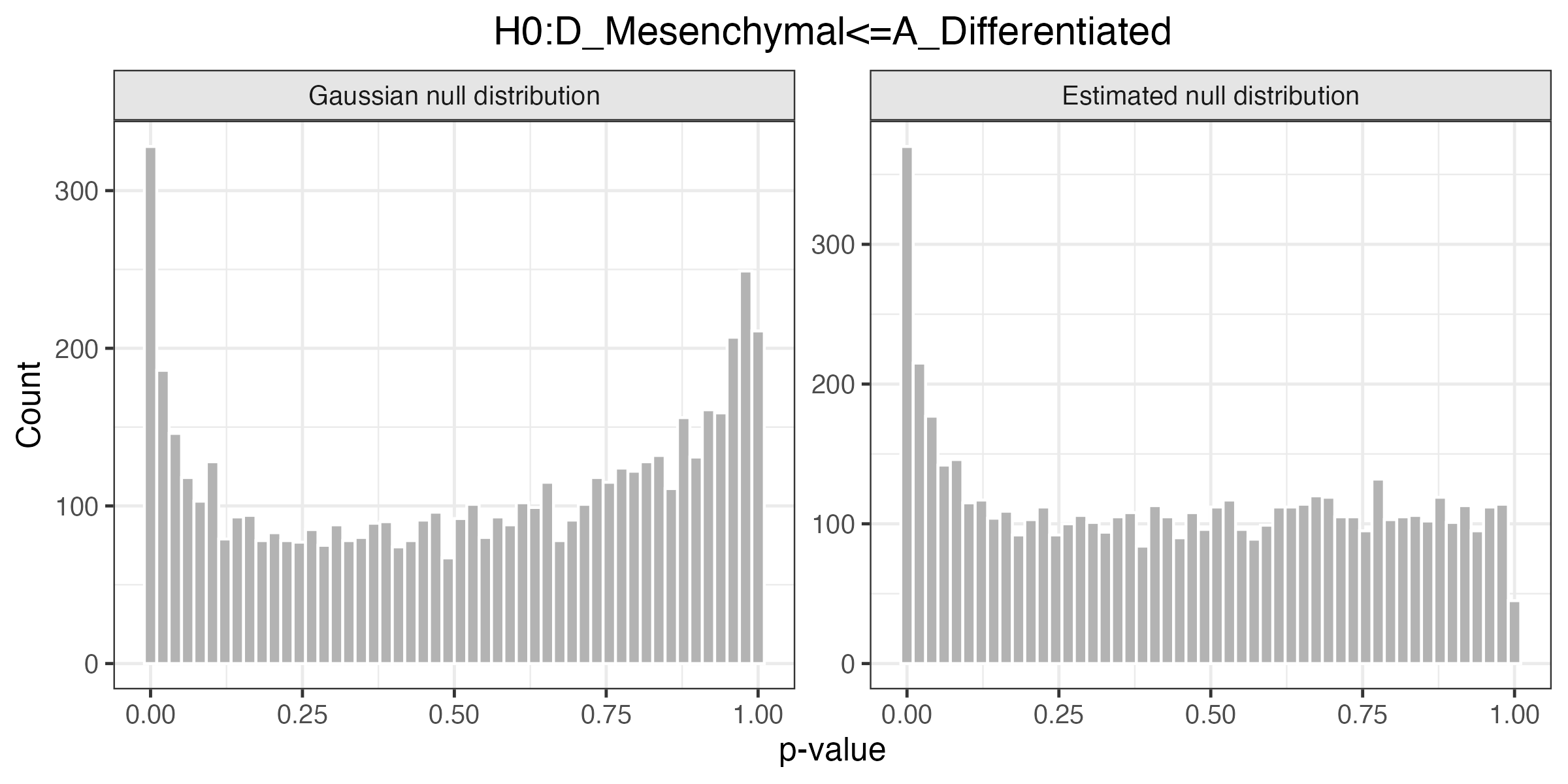}
    \caption{Histograms of the $p$-values computed under the standard Gaussian null distribution (left) and under our estimated null distribution (right), corresponding to the hypothesis tests $H_{0,i}: \mu_{D,i} \le \mu_{A,i}$.}
    \label{fig:9}
\end{figure}

\begin{figure}[!htb]
    \centering
    \includegraphics[width=0.4\linewidth]{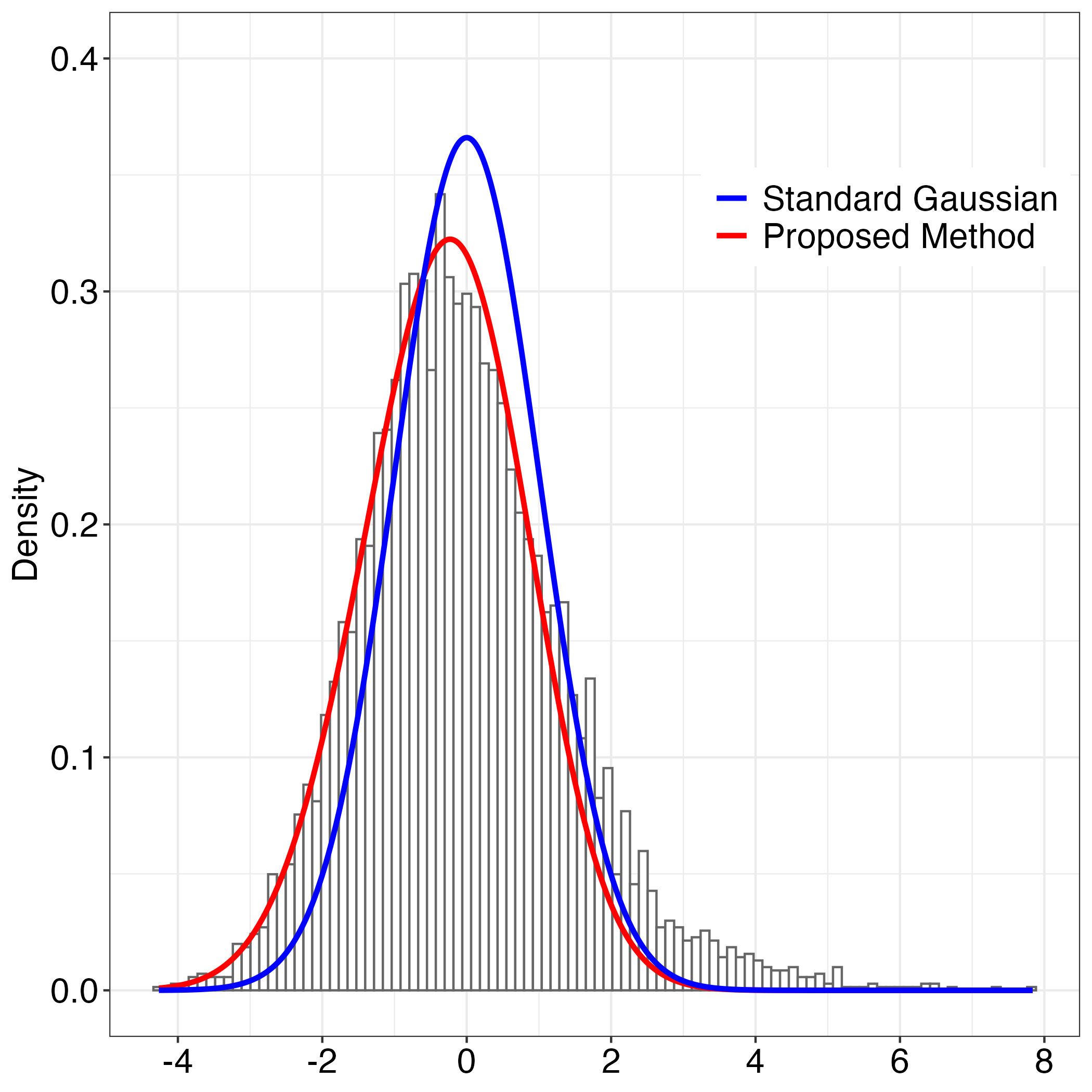}
    \caption{Comparison between the scaled null density estimated by our proposed method (red line) and the standard Gaussian distribution scaled by the same factor (blue line), corresponding to the hypothesis tests $H_{0,i}: \mu_{D,i} \le \mu_{A,i}$.}
    \label{fig:10}
\end{figure}

\section{Conclusion and Discussion}\label{sec:5}
In one-sided multiple testing problems, conservative $p$-values naturally arise and can lead to substantial power loss for standard procedures such as BH and Storey--BH. To address this issue, we propose a data-driven approach that estimates the null distribution $F_0$ and constructs $p$-values based on the estimated null, rather than relying on the overly conservative standard Gaussian null. Our method introduces three families of prior distributions for $M_0$, together with estimation algorithms tailored to each choice. Importantly, the resulting $p$-values can be analyzed using conventional multiple testing procedures without any modification. The key distinction of our approach is that it improves power by modifying the $p$-value construction step, rather than altering the multiple testing procedure itself. Through extensive simulation studies and a real data application, we demonstrate that the proposed method consistently outperforms existing approaches in terms of power, particularly in settings where a large fraction of $p$-values are conservative.

We conclude with a discussion of limitations and directions for future research. 
First, it is necessary to establish the (asymptotic) validity of multiple testing procedures, such as the BH or Storey--BH procedure, when applied to the empirical Bayes p-values. Since the empirical Bayes p-values share a common estimated prior distribution, they are not independent; consequently, the validity of multiple testing procedures applied to these p-values is not immediately guaranteed.

Second, as noted by a reviewer, the hypotheses are often dependent in many practical applications. In such cases, ignoring the dependency structure among hypotheses can adversely affect both the estimation of the null proportion and the FDR control---a limitation shared by the proposed method as well as existing methods. Addressing these challenges rigorously remains a non-trivial task and constitutes an important direction for future research.

Finally, there are many other sources of composite null hypotheses beyond the one-sided testing framework considered in this paper. For instance, Ryan Sun et al. \cite{sun2025testing} consider a union null hypothesis, which itself constitutes a form of composite null. In that setting, they propose csmGmm, a method that operates in a multivariate framework by fitting a constrained Gaussian mixture model designed to prevent incongruities between rankings based on the local false discovery rate and those based on classical frequentist significance. Extending our approach to accommodate the broader landscape of composite null structures arising from such diverse settings represents another promising avenue for future research.

\section*{Acknowledgement}
The authors declare that there are no conflicts of interest.

\bibliographystyle{apalike}
\bibliography{ref}

\appendix
\numberwithin{equation}{section}
\numberwithin{table}{section}
\numberwithin{figure}{section}

\numberwithin{theorem}{section}
\numberwithin{lemma}{section}
\numberwithin{proposition}{section}
\numberwithin{corollary}{section}

\section{Proof of Lemma~\ref{lemma1.1}}\label{appen_a}

\begin{proof}[Proof of Lemma~\ref{lemma1.1}]
    For a fixed null index $j$, we want to show that $p_j = 1 - \Phi(Z_j)$ is uniformly valid. A sufficient condition for the uniform validity of $p_j$ is the convexity of its c.d.f..
    Let $H$ denote the c.d.f. of $p_j$ and 
    let $F_0$ denote the marginal distribution of the null $Z_j$; that is,
    \begin{equation*}
        F_0(z) = \int_{-\infty}^z \int \phi(t; \mu, 1) ~dM_0(\mu) ~dt.
    \end{equation*}
    Then, we have
    \begin{equation*}
        H(x) 
        = \mathbb{P}(p_j \leq x) 
        = \mathbb{P}(1 - \Phi(Z_j) \leq x) 
        = \mathbb{P}(Z_j \geq \Phi^{-1}(1 - x)) 
        = 1 - F_0(\Phi^{-1}(1 - x)).    
    \end{equation*}
    By the chain rule and the inverse function theorem,
    \begin{equation*}
        h(x) 
        = H^\prime(x) 
        = F_0^{\prime}(\Phi^{-1}(1 - x)) \cdot (\Phi^{-1})^\prime (1-x) 
        = \frac{F_0^\prime(\Phi^{-1}(1 - x))}{\phi(\Phi^{-1}(1 - x))}.    
    \end{equation*}
    Since the Gaussian location family has the monotone likelihood ratio (MLR) property, $h(x)$ is a non-decreasing function of $x$, which implies that $H$ is convex. This completes the proof of the first part.
    
    For the second part, since $p_j$ is conservative whenever $M_0$ is not degenerate at zero, combining this with the result above we can establish the desired result.
\end{proof}

\section{Proof of Lemma~\ref{lemma2.1}}\label{appen_b}
\begin{proof}[Proof of Lemma~\ref{lemma2.1}]
    In this proof, we consider the case when $ \sigma^2 = 1 $. The result can be derived in a similar manner when $ \sigma^2 \neq 1 $. We follow the computation given in \cite{o1976bayes}.
    
    By the definition of $\tilde{M}_0$, we can express
    \begin{align*}
        f_0(z) 
        &= \int \phi(z-\mu) \, d\tilde{M}_0(\mu) \\
        &= \int_{-\infty}^0 \phi(z) 
            \exp\!\left\{ z\mu - \tfrac{1}{2}\mu^2 \right\} \cdot
            \frac{1}{\Phi \!\left(-\tfrac{\mu_0}{\sigma_0}\right)} 
            \frac{1}{\sqrt{2\pi}\sigma_0} \,
            \exp\!\left\{ -\tfrac{1}{2\sigma_0^2}(\mu - \mu_0)^2 \right\} d\mu \\
        &= \phi(z) 
           \frac{1}{\Phi \!\left(-\tfrac{\mu_0}{\sigma_0}\right)} 
           \frac{1}{\sqrt{2\pi}\sigma_0} 
           \underbrace{\int_{-\infty}^0 
           \exp\!\left\{ z\mu - \tfrac{1}{2}\mu^2 - \tfrac{1}{2\sigma_0^2}(\mu - \mu_0)^2 \right\} d\mu}_{\eqqcolon A}.
    \end{align*}
    To proceed, we examine the term $A$ in more detail.
    \begin{align*}
        A 
        &= \int_{-\infty}^0 
            \exp\!\left\{ z\mu - \tfrac{1}{2}\mu^2 
            - \tfrac{1}{2\sigma_0^2}\bigl(\mu^2 - 2\mu_0\mu + \mu_0^2\bigr) \right\} d\mu \\
        &= \exp\!\left\{ -\tfrac{\mu_0^2}{2\sigma_0^2} \right\} 
           \int_{-\infty}^0 
            \exp\!\left\{ -\tfrac{1}{2}\!\left( 1+\tfrac{1}{\sigma_0^2} \right)
            \Biggl( \mu^2 - \tfrac{2\sigma_0^2}{\sigma_0^2+1}\!\left( z + \tfrac{\mu_0}{\sigma_0^2}\right)\mu \Biggr) \right\} d\mu \\
        &= \exp\!\left\{ -\tfrac{\mu_0^2}{2\sigma_0^2} \right\} 
           \sqrt{2\pi}\,\tilde{\sigma} 
           \exp\!\left\{ \tfrac{1}{2\tilde{\sigma}^2}\,\tilde{\sigma}^4
            \left( z + \tfrac{\mu_0}{\sigma_0^2} \right)^2 \right\} 
            \frac{1}{\sqrt{2\pi}\,\tilde{\sigma}} \\
        &\quad\quad \times\int_{-\infty}^0 \exp\!\left\{ -\tfrac{1}{2\tilde{\sigma}^2}
        \Bigl( \mu - \tilde{\sigma}^2\!\left( z + \tfrac{\mu_0}{\sigma_0^2} \right) \Bigr)^2 \right\} d\mu  \\
        &= \exp\!\left\{ -\tfrac{\mu_0^2}{2\sigma_0^2} \right\} 
           \sqrt{2\pi}\,\tilde{\sigma} 
           \exp\!\left\{ \tfrac{1}{2}\tilde{\sigma}^2
            \left( z + \tfrac{\mu_0}{\sigma_0^2} \right)^2 \right\} 
            \Phi\!\left( -\tilde{\sigma}\!\left( z + \tfrac{\mu_0}{\sigma_0^2} \right) \right),
    \end{align*}
    where $\tilde{\sigma}^2 = \dfrac{\sigma_0^2}{1+\sigma_0^2}$. Using the expression obtained for $A$, $f_0(z)$ can be rewritten as
    \begin{align*}
        f_0(z) 
        &= \frac{\tilde{\sigma}}{\sigma_0} 
           \frac{1}{\Phi\!\left( -\tfrac{\mu_0}{\sigma_0} \right)} 
           \Phi\!\left( -\tilde{\sigma}\!\left( z + \tfrac{\mu_0}{\sigma_0^2} \right) \right) 
           \exp\!\left\{ -\tfrac{\mu_0^2}{2\sigma_0^2} \right\} 
           \exp\!\left\{ \tfrac{\tilde{\sigma}^2}{2} \left( z + \tfrac{\mu_0}{\sigma_0^2} \right)^2 \right\} 
            \phi(z) \\
        &= \frac{\tilde{\sigma}}{\sqrt{2\pi}\,\sigma_0} 
           \frac{1}{\Phi\!\left( -\tfrac{\mu_0}{\sigma_0} \right)}  
           \Phi\!\left( -\tilde{\sigma}\!\left( z + \tfrac{\mu_0}{\sigma_0^2} \right) \right) 
           \exp\!\left\{ -\tfrac{\mu_0^2}{2\sigma_0^2} \right\} 
           \underbrace{\exp\!\left\{ \tfrac{\tilde{\sigma}^2}{2} \left( z + \tfrac{\mu_0}{\sigma_0^2} \right)^2 \right\} 
            \exp\!\left\{ -\tfrac{z^2}{2} \right\}}_{\eqqcolon B}.
    \end{align*}
    Consider the term denoted by $B$.
    \begin{align*}
        B 
        &= \exp\!\left\{ -\tfrac{1}{2}z^2 
            + \tfrac{1}{2}\tilde{\sigma}^2 z^2 
            + \tilde{\sigma}^2 \tfrac{\mu_0}{\sigma_0^2} z 
            + \tfrac{\tilde{\sigma}^2}{2} \tfrac{\mu_0^2}{\sigma_0^4} \right\} \\
        &= \exp\!\left\{ -\tfrac{1}{2(1+\sigma_0^2)} 
            (z - \mu_0)^2 
            + \tfrac{\mu_0^2}{2\sigma_0^2} 
             \right\} \\
        &= \phi\!\left(z \,;\, \mu_0,\, 1+\sigma_0^2 \right) 
           \exp\!\left\{ \tfrac{\mu_0^2}{2\sigma_0^2} \right\}
           \sqrt{2\pi(1+\sigma_0^2)}.
    \end{align*}
    Using the expression for $B$, $f_0(z)$ can be rewritten as
    \begin{align*}
        f_0(z) 
        &= \frac{1}{\Phi\!\left( -\tfrac{\mu_0}{\sigma_0} \right)} 
           \phi\!\left(z ; \mu_0, 1+\sigma_0^2 \right) 
           \Phi\!\left( -\tfrac{\sigma_0^2 z + \mu_0}{\sigma_0 \sqrt{1 + \sigma_0^2}} \right)\\
        &= \dfrac{1}{\omega \Phi(\zeta)}\phi \left( \dfrac{z-\mu_0}{\omega}\right) \Phi\left( \alpha_0(\zeta) + \alpha \dfrac{z - \mu_0}{\omega} \right),
    \end{align*}
    where $\omega = \sqrt{1+\sigma_0^2},\; \zeta = -\dfrac{\mu_0}{\sigma_0},\; \alpha = -\sigma_0,$ and $\alpha_0(\zeta) = \zeta\sqrt{\alpha^2 + 1}$.
\end{proof}

\section{Proof of Theorems~\ref{theorem2.1} and~\ref{theorem2.3}}\label{appen_c}
We first establish the auxiliary lemmas needed for the proof of Theorem~\ref{theorem2.1}.

\begin{lemma}[Argmax Theorem]\label{lemmaC.1}
Let $\Theta \subset \mathbb{R}$ be compact. Let $M_n:\Theta\to\mathbb{R}$ and $M:\Theta\to\mathbb{R}$ be real-valued functions. Suppose that
\begin{itemize}
    \item[(i)] $\hat{\theta}_n \in \arg\max_{\theta\in\Theta} M_n(\theta)$ for each $n$,
    \item[(ii)] $\sup_{\theta\in\Theta}|M_n(\theta)-M(\theta)|\xrightarrow{p}0$ as $n\to\infty$,
    \item[(iii)] $\theta_0\in\Theta$ is the unique maximizer of $M(\theta)$ in the sense that for every $\epsilon>0$,
    \[
    \sup_{\theta\in\Theta:|\theta-\theta_0|\ge\epsilon}M(\theta)<M(\theta_0).
    \]
\end{itemize}
Then, $\hat{\theta}_n\xrightarrow{p}\theta_0$ as $n\to\infty$.
\end{lemma}

\begin{proof}
Fix $\epsilon>0$ and define
\[
A_\epsilon=\{\theta\in\Theta:|\theta-\theta_0|\ge\epsilon\}.
\]
By assumption (iii), there exists $\eta_\epsilon>0$ such that
\begin{equation}\label{eqn:lemmaC.1_1}
    \eta_\epsilon \coloneqq M(\theta_0)-\sup_{\theta\in A_\epsilon}M(\theta) > 0.
\end{equation}
Define the event
\begin{equation}\label{eqn:lemmaC.1_2}
    E_n= E_n(\epsilon) = \left\{\sup_{\theta\in\Theta}|M_n(\theta)-M(\theta)|<\frac{\eta_\epsilon}{3}\right\},
\end{equation}
so that $\mathbb{P}(E_n)\to 1$ as $n\to\infty$ by assumption (ii). 

On the event $E_n$,
\begin{equation}\label{eqn:lemmaC.1_3}
    M_n(\theta_0)> M(\theta_0)-\frac{\eta_\epsilon}{3},
\end{equation}
and for any $\theta\in A_\epsilon$,
\begin{equation}\label{eqn:lemmaC.1_4}
    M_n(\theta)
    <
    M(\theta)+\frac{\eta_\epsilon}{3}
    \le
    M(\theta_0)-\eta_\epsilon+\frac{\eta_\epsilon}{3}
    =
    M(\theta_0)-\frac{2\eta_\epsilon}{3},
\end{equation}
where the first inequality follows from \eqref{eqn:lemmaC.1_2} and the second from \eqref{eqn:lemmaC.1_1}. Combining \eqref{eqn:lemmaC.1_3} and \eqref{eqn:lemmaC.1_4} yields $M_n(\theta)<M_n(\theta_0)$ for every $\theta\in A_\epsilon$ on $E_n$, so by compactness of $\Theta$ and assumption (i), $\hat{\theta}_n\in\Theta\setminus A_\epsilon$ on $E_n$, i.e.,
\[
E_n\subset\{|\hat{\theta}_n-\theta_0|<\epsilon\}.
\]
Therefore, we obtain 
\[
\mathbb{P}(|\hat{\theta}_n-\theta_0|\ge\epsilon)\le\mathbb{P}(E_n^c).
\]
Since $\mathbb{P}(E_n^c) \to 0$ by construction for the given $\epsilon$, and since $\epsilon > 0$ was arbitrary, it follows that $\hat{\theta}_n \xrightarrow{p} \theta_0$.
\end{proof}

\begin{lemma}[Uniform Law of Large Numbers via Bracketing]\label{lemmaC.2}
Let $\mathcal{F}=\{f_\theta:\theta\in\Theta\}$ be a class of measurable functions. Suppose that for every $\epsilon>0$, the $\epsilon$-bracketing number $N_{[\,]}(\epsilon,\mathcal{F},L^1(P))$ is finite, i.e., there exist finitely many brackets $[l_j,u_j]$ with $\mathbb{E}[u_j(X)-l_j(X)]<\epsilon$ such that for every $f\in\mathcal{F}$, there exists $j$ with $l_j(x)\le f(x)\le u_j(x)$ for all $x$. Then, $\mathcal{F}$ is Glivenko--Cantelli, i.e.,
\[
\sup_{f\in\mathcal{F}}\left|\frac{1}{n}\sum_{i=1}^n f(X_i)-\mathbb{E}[f(X)]\right|\xrightarrow{p}0.
\]
\end{lemma}

\begin{proof}
Fix $\epsilon>0$. By assumption, there exist finitely many brackets $[l_1,u_1],\dots,[l_K,u_K]$ such that:
\begin{itemize}
    \item[(i)] For every $f\in\mathcal{F}$, there exists $j\in\{1,\dots,K\}$ with $l_j(x)\le f(x)\le u_j(x)$ for all $x$.
    \item[(ii)] $\mathbb{E}[u_j(X)-l_j(X)]<\epsilon$ for each $j$.
\end{itemize}
For each $f\in\mathcal{F}$, let $j(f)$ denote the index of the bracket containing $f$. By conditions (i) and (ii), for any $f\in\mathcal{F}$,
\begin{equation}\label{eqn:lemmaC.2_1}
\begin{split}
    \frac{1}{n}\sum_{i=1}^n f(X_i)-\mathbb{E}[f(X)]
    &\le
    \frac{1}{n}\sum_{i=1}^n u_{j(f)}(X_i)-\mathbb{E}[l_{j(f)}(X)]\\
    &=
    \left(\frac{1}{n}\sum_{i=1}^n u_{j(f)}(X_i)-\mathbb{E}[u_{j(f)}(X)]\right)\\
    &\qquad+
    \left(\mathbb{E}[u_{j(f)}(X)]-\mathbb{E}[l_{j(f)}(X)]\right)\\
    &\le 
    \max_{1\le j\le K}\left|\frac{1}{n}\sum_{i=1}^n u_j(X_i)-\mathbb{E}[u_j(X)]\right| + \epsilon,
\end{split}
\end{equation}
and by a symmetric argument,
\begin{equation}\label{eqn:lemmaC.2_2}
    \mathbb{E}[f(X)] - \frac{1}{n} \sum_{i = 1}^n f(X_i) \leq 
    \max_{1\le j\le K}\left|\frac{1}{n}\sum_{i=1}^n l_j(X_i)-\mathbb{E}[l_j(X)]\right|+\epsilon.
\end{equation}
Combining \eqref{eqn:lemmaC.2_1} and \eqref{eqn:lemmaC.2_2} and taking the supremum over $f\in\mathcal{F}$,
\begin{align*}
    \sup_{f\in\mathcal{F}}\left|\frac{1}{n}\sum_{i=1}^n f(X_i)-\mathbb{E}[f(X)]\right| 
    &\leq 
    \max_{1\le j\le K}\left|\frac{1}{n}\sum_{i=1}^n u_j(X_i)-\mathbb{E}[u_j(X)]\right|\\
    &\qquad+ 
    \max_{1\le j\le K}\left|\frac{1}{n}\sum_{i=1}^n l_j(X_i)-\mathbb{E}[l_j(X)]\right| + \epsilon.
\end{align*}
Since the maxima are over finitely many $j$, the law of large numbers gives
\[
\max_{1\le j\le K}\left|\frac{1}{n}\sum_{i=1}^n u_j(X_i)-\mathbb{E}[u_j(X)]\right| + \max_{1\le j\le K}\left|\frac{1}{n}\sum_{i=1}^n l_j(X_i)-\mathbb{E}[l_j(X)]\right| = o_p(1).
\]
Since $\epsilon>0$ was arbitrary, we conclude
\[
\sup_{f\in\mathcal{F}}\left|\frac{1}{n}\sum_{i=1}^n f(X_i)-\mathbb{E}[f(X)]\right|\xrightarrow{p}0. \qedhere
\]
\end{proof}

With the above preparations, we are now ready to prove Theorem~\ref{theorem2.1}.

\begin{proof}[Proof of Theorem~\ref{theorem2.1}]
Define
\[
g_\mu(z)=\log f_\mu^T(z)
=
-\frac{1}{2}\log(2\pi)
-\frac{1}{2}(z-\mu)^2
-\log\Phi(\xi-\mu)
+\log\mathbb{I}(z\le\xi),
\]
and the population criterion
\[
M(\mu)=\mathbb{E}_{\mu_0}[g_\mu(Z)].
\]
We verify assumptions (ii) and (iii) of Lemma~\ref{lemmaC.1}.

\medskip
\emph{Step 1. Uniform convergence.}
We show that
\[
\sup_{\mu\in\Theta}\left|M_n(\mu)-M(\mu)\right| 
= \sup_{\mu\in\Theta}\left|\frac{1}{n}\sum_{i=1}^n g_\mu(Z_i) - \mathbb{E}_{\mu_0}[g_\mu(Z)]\right|
\xrightarrow{p}0.
\]
Since $\Theta$ is compact, it is bounded, so there exist $A,B\in\mathbb{R}$ such that $\Theta\subset[A,B]$. The map $\mu\mapsto g_\mu(z)$ is continuously differentiable with
\[
\frac{\partial}{\partial\mu}g_\mu(z)
=
z-\mu+\frac{\phi(\xi-\mu)}{\Phi(\xi-\mu)}.
\]
Since $\mu\in\Theta\subset[A,B]$, the term $-\mu+\phi(\xi-\mu)/\Phi(\xi-\mu)$ is bounded uniformly over $\Theta$, so there exists a constant $C>0$ such that
\[
\left|\frac{\partial}{\partial\mu}g_\mu(z)\right|\le |z|+C.
\]
By the mean value theorem, for any $\mu,\mu'\in\Theta$,
\[
|g_\mu(z)-g_{\mu'}(z)|\le|\mu-\mu'|\cdot(|z|+C).
\]
Setting $L(z)=|z|+C$, we have $\mathbb{E}_{\mu_0}[L(Z)]<\infty$ since $X$ follows a truncated normal distribution. Fix $\epsilon>0$ and choose $\delta>0$ such that $2\delta\,\mathbb{E}_{\mu_0}[L(Z)]=\epsilon$. Cover $[A,B]$ by finitely many intervals of length at most $\delta$ with representative points $\mu_1,\dots,\mu_K$, and define the brackets
\[
l_j(z)=g_{\mu_j}(z)-\delta L(z),
\qquad
u_j(z)=g_{\mu_j}(z)+\delta L(z).
\]
Every $g_\mu$ with $\mu\in\Theta$ is sandwiched within one of these brackets, and
\[
\mathbb{E}_{\mu_0}[u_j(Z)-l_j(Z)]=2\delta\,\mathbb{E}_{\mu_0}[L(Z)]=\epsilon.
\]
Hence, the class $\mathcal{G}=\{g_\mu:\mu\in\Theta\}$ has a finite $\epsilon$-bracketing number in $L^1(P_{\mu_0})$. By Lemma~\ref{lemmaC.2}, $\mathcal{G}$ is Glivenko--Cantelli, and thus
\[
\sup_{\mu\in\Theta}\left|M_n(\mu)-M(\mu)\right|\xrightarrow{p}0.
\]

\medskip
\emph{Step 2. Unique maximizer.} 
Next, we show that $\mu_0$ is the unique maximizer of $M(\mu)$ in the sense that for every $\epsilon > 0$, 
\[
\sup_{\mu\in\Theta:|\mu-\mu_0|\ge\epsilon}M(\mu)<M(\mu_0).
\]
For any $\mu\in\Theta$,
\[
M(\mu_0)-M(\mu)
=\mathbb{E}_{\mu_0}\!\left[\log\frac{f_{\mu_0}^T(Z)}{f_\mu^T(Z)}\right]
=\mathrm{KL}(f_{\mu_0}^T\,\|\,f_\mu^T)\ge 0,
\]
so $M(\mu)\le M(\mu_0)$ for all $\mu\in\Theta$. If equality holds, then $\mathrm{KL}(f_{\mu_0}^T\,\|\,f_\mu^T)=0$, which implies $f_{\mu_0}^T(z)=f_\mu^T(z)$ for almost every $z<\xi$, i.e.,
\[
-\frac{1}{2}(z-\mu_0)^2-\log\Phi(\xi-\mu_0)
=
-\frac{1}{2}(z-\mu)^2-\log\Phi(\xi-\mu).
\]
Since both sides are quadratic in $z$ and agree on an interval, comparing coefficients of $z$ yields $\mu=\mu_0$. Hence $M$ is uniquely maximized at $\mu_0$. Since $\Theta$ is compact and $M$ is continuous, this uniqueness implies that for every $\epsilon>0$,
\[
\sup_{\mu\in\Theta:|\mu-\mu_0|\ge\epsilon}M(\mu)<M(\mu_0).
\]

The conclusion $\hat{\mu}_n\xrightarrow{p}\mu_0$ now follows from Lemma~\ref{lemmaC.1}.
\end{proof}

We now turn to the proof of Theorem~\ref{theorem2.3}. To this end, we first establish the following auxiliary lemmas.

\begin{lemma}[Properties of the Skew-Normal Density]\label{lemmaC.3}
Let $f_\sigma$ be the skew-normal density defined as
\[
f_\sigma(z)
=
\frac{2}{\sqrt{1+\sigma^2}}
\phi\!\left(\frac{z}{\sqrt{1+\sigma^2}}\right)
\Phi\!\left(-\frac{\sigma z}{\sqrt{1+\sigma^2}}\right),
\qquad \sigma>0.
\]
Then, the following hold:
\begin{itemize}
    \item[(i)] $f_\sigma$ is real analytic on $\mathbb{R}$ for every $\sigma>0$.
    \item[(ii)] The map $\sigma\mapsto f_\sigma(z)$ is continuously differentiable, and for any compact interval $[A,B]\subset(0,\infty)$, there exists a constant $C>0$ such that
    \[
    \left|\frac{\partial}{\partial\sigma}\log f_\sigma(z)\right|\le C(1+z^2)
    \quad\text{for all }\sigma\in[A,B]\text{ and }z\in\mathbb{R}.
    \]
    \item[(iii)] The mean of $f_\sigma$ is $\mathbb{E}_\sigma[Z]=-\sigma\sqrt{2/\pi}$, which is strictly monotone in $\sigma>0$.
\end{itemize}
\end{lemma}
\begin{proof}
\textit{(i)} Since $\phi$ and $\Phi$ are real analytic on $\mathbb{R}$, and real analyticity is preserved under composition with linear functions, products, and scalar multiplication, $f_\sigma$ is real analytic on $\mathbb{R}$.

\medskip
\noindent\textit{(ii)} We can write
\[
\log f_\sigma(z)
=
\log 2 - \frac{1}{2}\log(1+\sigma^2) - \frac{1}{2}\log(2\pi)
-\frac{z^2}{2(1+\sigma^2)}
+\log\Phi\!\left(-\frac{\sigma z}{\sqrt{1+\sigma^2}}\right),
\]
so that
\[
\frac{\partial}{\partial\sigma}\log f_\sigma(z)
=
-\frac{\sigma}{1+\sigma^2}
+\frac{\sigma z^2}{(1+\sigma^2)^2}
-\frac{z}{(1+\sigma^2)^{3/2}}\cdot\frac{\phi\!\left(-\frac{\sigma z}{\sqrt{1+\sigma^2}}\right)}{\Phi\!\left(-\frac{\sigma z}{\sqrt{1+\sigma^2}}\right)}.
\]
Since each term on the right-hand side is a continuous function of $\sigma$, the map $\sigma\mapsto\partial\log f_\sigma(z)/\partial\sigma$ is continuous, and hence $\sigma\mapsto\partial f_\sigma(z)/\partial\sigma = \partial\log f_\sigma(z)/\partial\sigma\cdot f_\sigma(z)$ is continuous as well. This establishes that $\sigma\mapsto f_\sigma(z)$ is continuously differentiable.

We now establish the bound $|\partial\log f_\sigma(z)/\partial\sigma|\le C(1+z^2)$. Since $\sigma\in[A,B]$, there exists a constant $C_0>0$ such that
\begin{equation}\label{eqn:lemma2_1}
    \left|-\frac{\sigma}{1+\sigma^2}+\frac{\sigma z^2}{(1+\sigma^2)^2}\right|
    \le
    C_0(1+z^2).
\end{equation}
For the remaining term, the Mills ratio bound $\phi(t)/\Phi(-t)\le|t|+\sqrt{2/\pi}$ gives
\begin{equation}\label{eqn:lemma2_2}
    \left|\frac{z}{(1+\sigma^2)^{3/2}}\cdot\frac{\phi\!\left(-\frac{\sigma z}{\sqrt{1+\sigma^2}}\right)}{\Phi\!\left(-\frac{\sigma z}{\sqrt{1+\sigma^2}}\right)}\right|
    \le
    \frac{|\sigma|z^2}{(1+\sigma^2)^{2}}+\sqrt{\frac{2}{\pi}}\frac{|z|}{(1+\sigma^2)^{3/2}}
    \le
    C_1(1+z^2)
\end{equation}
for some constant $C_1>0$, where the last inequality uses $|z|\le 1+z^2$ and $\sigma\in[A,B]$. Combining \eqref{eqn:lemma2_1} and \eqref{eqn:lemma2_2} yields $\left|\partial\log f_\sigma(z)/\partial\sigma\right|\le C(1+z^2)$ for some $C>0$.

\medskip
\noindent\textit{(iii)} This follows directly from the standard moment formula for the skew-normal distribution.
\end{proof}

\begin{lemma}\label{lemmaC.4}
For any compact interval $[A,B]\subset(0,\infty)$, the map $\sigma\mapsto\log F_\sigma(\xi)$ is continuously differentiable on $[A,B]$, and
\[
\sup_{\sigma\in[A,B]}\left|\frac{\partial}{\partial\sigma}\log F_\sigma(\xi)\right|<\infty.
\]
\end{lemma}

\begin{proof}
Using $\Phi(\cdot)\le 1$, $\sigma\in[A,B]$, and the fact that $\phi$ is decreasing in $|z|/\sqrt{1+\sigma^2}$, we obtain
\begin{equation}\label{eqn:lemmaC.4_1}
    f_\sigma(z)
    \le
    \frac{2}{\sqrt{1+A^2}}\phi\!\left(\frac{z}{\sqrt{1+B^2}}\right)
    =:h(z).
\end{equation}
By the mean value theorem, \eqref{eqn:lemmaC.4_1}, and Lemma~\ref{lemmaC.3}(ii), there exists a constant $C>0$ such that
\[
\left|\frac{f_{\sigma+t}(z) - f_\sigma(z)}{t}\right| 
\leq 
\sup_{\sigma \in [A, B]} \left|\frac{\partial}{\partial\sigma}\log f_\sigma(z)\right|\cdot f_\sigma(z)
\le
C(1+z^2)h(z)
\quad\text{for all }\sigma\in[A,B].
\]
Since $\int C(1+z^2)h(z)\,dz<\infty$, the dominated convergence theorem justifies differentiation under the integral sign and simultaneously establishes differentiability of $\sigma\mapsto F_\sigma(\xi)$, with
\[
\frac{\partial}{\partial\sigma}F_\sigma(\xi)
=
\int_{-\infty}^\xi\frac{\partial}{\partial\sigma}f_\sigma(z)\,dz.
\]

Next, we show continuity of this derivative. To this end, let $\sigma_n\to\sigma$ in $[A,B]$. Since $\sigma\mapsto(\partial/\partial\sigma)f_\sigma(z)$ is continuous pointwise in $z$ by Lemma~\ref{lemmaC.3}(ii) and is dominated by the integrable function $C(1+z^2)h(z)$, the dominated convergence theorem gives
\[
\frac{\partial}{\partial\sigma_n}F_{\sigma_n}(\xi)
=
\int_{-\infty}^\xi\frac{\partial}{\partial\sigma_n}f_{\sigma_n}(z)\,dz
\to
\int_{-\infty}^\xi\frac{\partial}{\partial\sigma}f_\sigma(z)\,dz
=
\frac{\partial}{\partial\sigma}F_\sigma(\xi).
\]
Hence $\sigma\mapsto F_\sigma(\xi)$ is continuously differentiable on $[A,B]$.

It remains to pass from $F_\sigma(\xi)$ to $\log F_\sigma(\xi)$. Since $F_\sigma(\xi)=\mathbb{P}_\sigma(Z\le\xi)>0$ for every $\sigma\in[A,B]$ and $\sigma\mapsto F_\sigma(\xi)$ is continuous on the compact set $[A,B]$,
\begin{equation}\label{lemmaC.4_2}
    \inf_{\sigma\in[A,B]}F_\sigma(\xi)>0.       
\end{equation}
Since $\sigma\mapsto F_\sigma(\xi)$ is continuously differentiable and bounded away from zero, the chain rule implies that $\sigma\mapsto\log F_\sigma(\xi)$ is continuously differentiable.

Finally, since $[A,B]$ is compact and $\sigma\mapsto\partial F_\sigma(\xi)/\partial\sigma$ is continuous, $\sup_{\sigma\in[A,B]}|\partial F_\sigma(\xi)/\partial\sigma|<\infty$. Combined with \eqref{lemmaC.4_2}, we have
\[
\sup_{\sigma\in[A,B]}\left|\frac{\partial}{\partial\sigma}\log F_\sigma(\xi)\right|
\le
\frac{\sup_{\sigma\in[A,B]}|\partial F_\sigma(\xi)/\partial\sigma|}{\inf_{\sigma \in [A, B]} F_\sigma(\xi)}
<\infty. \qedhere
\]
\end{proof}

We are now ready to prove Theorem~\ref{theorem2.3}.

\begin{proof}[Proof of Theorem~\ref{theorem2.3}]
Define
\[
g_\sigma(z)=\log f_\sigma^T(z)
=
\log f_\sigma(z)-\log F_\sigma(\xi),
\quad z<\xi,
\]
and the population criterion
\[
M(\sigma)=\mathbb{E}_{\sigma_0}[g_\sigma(Z)], \quad \sigma\in\Theta.
\]
We verify assumptions (ii) and (iii) of Lemma~\ref{lemmaC.1}.

\medskip
\emph{Step 1. Uniform convergence.}
Since $\Theta$ is compact, there exist $A,B>0$ such that $\Theta\subset[A,B]$. By Lemma~\ref{lemmaC.3}(ii) and Lemma~\ref{lemmaC.4}, there exists a constant $C>0$ such that
\[
\left|\frac{\partial}{\partial\sigma}g_\sigma(z)\right|
=
\left|\frac{\partial}{\partial\sigma}\log f_\sigma(z)-\frac{\partial}{\partial\sigma}\log F_\sigma(\xi)\right|
\le C(1+z^2).
\]
By the mean value theorem, for any $\sigma,\sigma'\in\Theta$,
\[
|g_\sigma(z)-g_{\sigma'}(z)|\le|\sigma-\sigma'|\cdot C(1+z^2).
\]
Setting $L(z)=C(1+z^2)$, we have $\mathbb{E}_{\sigma_0}[L(Z)]<\infty$ since $Z$ has finite second moment. Fix $\epsilon>0$ and choose $\delta>0$ such that $2\delta\,\mathbb{E}_{\sigma_0}[L(Z)]=\epsilon$. Cover $[A,B]$ by finitely many intervals of length at most $\delta$ with representative points $\sigma_1,\dots,\sigma_K$, and define the brackets
\[
l_j(z)=g_{\sigma_j}(z)-\delta L(z),
\qquad
u_j(z)=g_{\sigma_j}(z)+\delta L(z).
\]
Every $g_\sigma$ with $\sigma\in\Theta$ is sandwiched within one of these brackets, and
\[
\mathbb{E}_{\sigma_0}[u_j(Z)-l_j(Z)]=2\delta\,\mathbb{E}_{\sigma_0}[L(Z)]=\epsilon.
\]
Hence $\mathcal{G}=\{g_\sigma:\sigma\in\Theta\}$ has a finite $\epsilon$-bracketing number in $L^1(P_{\sigma_0})$, so by Lemma~\ref{lemmaC.2}, the uniform law of large numbers holds:
\[
\sup_{\sigma\in\Theta}|M_n(\sigma)-M(\sigma)|\xrightarrow{p}0.
\]

\medskip
\emph{Step 2. Unique maximizer.}
Next, for any $\sigma\in\Theta$,
\[
M(\sigma_0)-M(\sigma)
=\mathbb{E}_{\sigma_0}\!\left[\log\frac{f_{\sigma_0}^T(Z)}{f_\sigma^T(Z)}\right]
=\mathrm{KL}(f_{\sigma_0}^T\,\|\,f_\sigma^T)\ge 0,
\]
so $M(\sigma)\le M(\sigma_0)$ for all $\sigma\in\Theta$. If equality holds, then $\mathrm{KL}(f_{\sigma_0}^T\,\|\,f_\sigma^T)=0$, which implies $f_{\sigma_0}^T(z)=f_\sigma^T(z)$ for almost every $z<\xi$, i.e.,
\[
\frac{f_{\sigma_0}(z)}{F_{\sigma_0}(\xi)}
=
\frac{f_\sigma(z)}{F_\sigma(\xi)}
\quad\text{for a.e. }z<\xi.
\]
By Lemma~\ref{lemmaC.3}(i), $f_{\sigma_0}$ and $f_\sigma$ are real analytic, so the function
\[
h(z):=\frac{f_{\sigma_0}(z)}{F_{\sigma_0}(\xi)}-\frac{f_\sigma(z)}{F_\sigma(\xi)}
\]
is real analytic on $\mathbb{R}$. Since $h$ is continuous and vanishes almost everywhere on $(-\infty,\xi)$, 
it vanishes everywhere on $(-\infty,\xi)$. The Identity Theorem then implies $h\equiv 0$ on $\mathbb{R}$, i.e.,
\[
\frac{f_{\sigma_0}(z)}{F_{\sigma_0}(\xi)}=\frac{f_\sigma(z)}{F_\sigma(\xi)}\quad\text{for all }z\in\mathbb{R}.
\]
Integrating both sides over $\mathbb{R}$ and using the fact that $f_{\sigma_0}$ and $f_\sigma$ are probability densities gives
\[
\frac{1}{F_{\sigma_0}(\xi)}=\frac{1}{F_\sigma(\xi)},
\]
and hence $F_{\sigma_0}(\xi)=F_\sigma(\xi)$. Substituting back yields $f_{\sigma_0}(z)=f_\sigma(z)$ for all $z\in\mathbb{R}$.
By Lemma~\ref{lemmaC.3}(iii), the mean $\mathbb{E}_\sigma[Z]=-\sigma\sqrt{2/\pi}$ is strictly monotone in $\sigma$, so equal densities imply $\sigma=\sigma_0$. Hence $M$ is uniquely maximized at $\sigma_0$. Since $\Theta$ is compact and $M$ is continuous, for every $\epsilon>0$,
\[
\sup_{\sigma\in\Theta:|\sigma-\sigma_0|\ge\epsilon}M(\sigma)<M(\sigma_0).
\]

The conclusion $\hat{\sigma}_n\xrightarrow{p}\sigma_0$ now follows from Lemma~\ref{lemmaC.1}.
\end{proof}

\section{Implementation Details of Our Proposed Method}\label{appen_d}
All three algorithms (Algorithms~\ref{alg:1}--\ref{alg:3}) require the truncation point $\xi$ as input, which is related to the zero assumption. A smaller value of $\xi$ yields a truncated set $\mathcal{S}_0(\xi)$ that contains more pure null indices, whereas a larger value of $\xi$ results in a more contaminated truncated set since the probability of including non-null indices increases with $\xi$. The optimal choice of $\xi$ depends on the null distribution, the non-null distribution, and their proportions. In this paper, we adopt a heuristic choice for $\xi$. Specifically, we set $\xi$ to be the 85th percentile of the observed statistics, which is equivalent to truncating the top 15\% of the observations. This choice is motivated by the common assumption that the null proportion satisfies $\pi_0 \geq 0.9$ in most applications \citep{efron2004large}. Under this assumption, the non-null proportion is at most 10\% of the observations. If the effect size is sufficiently large, our choice of $\xi$ is expected to remove almost all non-null $Z_i$ while still retaining roughly 85\% of the data for null density estimation. A data-adaptive choice of $\xi$ is left for future research.

For Algorithm~\ref{alg:2}, we additionally specify the search bounds $\eta_{\min}$ and $\eta_{\max}$. Recall that to remove the positivity constraint on $\sigma_0$, we employ the reparameterization $\eta = \log(\sigma_0)$. The bounds $\eta_{\min}$ and $\eta_{\max}$ restrict the range of $\eta$ in order to apply Brent’s method. By default, we set $\eta_{\min}=-6$ and $\eta_{\max}=3$, which roughly correspond to restricting $\sigma_0$ to the interval $[0.0025,20]$. This range is sufficiently wide to cover reasonable values of $\sigma_0$.

For Algorithm~\ref{alg:3}, the number of mixture components $K$ and the support points $\{\mu_k\}_{k=1}^K$ must be specified. The number of components $K$ controls the flexibility of the function class; therefore, we recommend using a sufficiently large value such as $K \geq 50$. Given $K$, the support points $\{\mu_k\}_{k=1}^K$ must also be specified. By default, we place equally spaced support points between $\min(Z_i)$ and $0$. Specifically, without loss of generality, let $\mu_1 \leq \mu_2 \leq \cdots \leq \mu_K$. 
Then, $\mu_1 = \min(Z_i)$, $\mu_K = 0$, and $\{\mu_k\}_{k=1}^K$ are equally spaced on this interval.

\end{document}